\documentclass[12pt]{article}
\usepackage[a4paper, margin=1in]{geometry}
\usepackage[page]{appendix}
\usepackage{graphicx}
\usepackage{setspace}
\usepackage{enumitem}
\usepackage[font=footnotesize,margin=0in]{caption}

\usepackage{xcolor}
\usepackage{environ}

\usepackage{amsmath,mathtools}
\usepackage{amssymb}
\usepackage{dsfont}
\usepackage{amsthm}
\usepackage{tikz,tikz-qtree}

\usepackage[hidelinks,hypertexnames=false]{hyperref}
\usepackage[capitalise]{cleveref}

\newtheorem{lemma}{Lemma}
\newtheorem{proposition}{Proposition}
\newtheorem{theorem}{Theorem}

\theoremstyle{definition}
\newtheorem{definition}{Definition}

\theoremstyle{remark} 

\newtheorem{example}{Example}
\Crefname{observation}{Observation}{Observations}

\newcommand\reals{\mathbb R}

\newcommand\E{\mathbb E}

\let\originalleft\left
\let\originalright\right
\renewcommand{\left}{\mathopen{}\mathclose\bgroup\originalleft}
\renewcommand{\right}{\aftergroup\egroup\originalright}

\DeclareMathOperator*{\argmax}{arg\,max}

\usepackage[backend=biber,authordate,natbib,useprefix=true]{biblatex-chicago}
\begin{document}

\title{Comparative Statics of Information Acquisition and Risk Aversion\thanks{Cabrales gratefully acknowledges funding from AEI under grants PID2024-156629NB-I00 and CEX2021-001181-M. Curello acknowledges support from the German Research Foundation (DFG) through CRC TR 224 (Project B02).  Gossner acknowledges support from the French National Research Agency (ANR), ``Investissements d'Avenir'' (ANR-11-IDEX-0003/LabEx Ecodec/ANR-11-LABX-0047) and from UKRI (SInfoNiA). Serrano acknowledges research support from Fundaci\'on ONCE, and thanks Universidad Carlos III for its warm hospitality.  ChatGPT has been used for tightening proofs and exposition, as well as in passages of related literature, which have been fully checked by the authors. All remaining errors are our own.}
\author{\normalsize Antonio Cabrales, Gregorio Curello, Olivier Gossner, and Roberto Serrano\thanks{Cabrales: Department of Economics, Universidad Carlos III, \texttt{antonio.cabrales@uc3m.es}; Curello: Department of Economics, University of Mannheim, \texttt{gregorio.curello@uni-mannheim.de}; 
Gossner: CNRS -- \'Ecole Polytechnique and Department of Mathematics, London School of
Economics, \texttt{ogossner@gmail.com}; Serrano: Department of Economics, Brown University, \texttt{roberto\_serrano@brown.edu}.}}}

\date{September 2026}

\maketitle

\noindent \textbf{Abstract:} This paper studies how willingness to pay for information depends on risk aversion. We model a decision maker who faces background risk and can acquire information before choosing from a menu of assets. We show that the comparative statics depend on two factors. The first is whether available assets constitute an investment menu, whose payoffs are procyclical with background wealth, or an insurance menu, whose payoffs are countercyclical. The second factor is the tail geometry of background risk. We show that willingness to pay for information decreases with risk aversion for investment menus when the density of background risk is log-concave, and that it increases with risk aversion for insurance menus when background risk is downward-log-convex. The proofs compare the distributions of terminal wealth with and without information. They develop new aggregation arguments for state-dependent single-crossing comparisons. We also construct reversals under strictly log-convex tails for investment menus and super-exponential left tails for insurance menus.



\newpage

\section{Introduction}\label{sec:introduction}

Information enables economic agents to adapt financial decisions to changing economic circumstances. It shapes both investment decisions \parencite{GrossmanStiglitz1980} and insurance choices \parencite{RothschildStiglitz1976}. When information acquisition is itself a choice, an important question is which agents are most willing to acquire information, and in particular how this willingness depends on their attitudes toward risk.

The answer depends on the use of this information. Less risk-averse investors may value information when it helps choose attractive investment opportunities. More risk-averse decision makers, on the other hand, may value information that lets them select insurance when it is most needed. We ask the following comparative statics question: under what conditions does greater risk aversion increase or decrease willingness to pay for information? 

We study this question in the presence of background risk. Before deciding whether to acquire information or select an asset, the decision maker already faces uncertainty about her initial wealth. She has a chance to buy a Blackwell experiment. After observing a signal realization, she can choose to purchase an asset. Information has a fixed price, which moves wealth down in every state, but it also allows the agent to choose an asset contingent on the signal. The value of information depends on how these two movements affect the final distribution of wealth.

In our main results, we show that the direction of the comparative statics depends both on the type of assets available and on the tail geometry of background risk. We distinguish investment menus, whose payoffs comove positively with background wealth, from insurance menus, whose payoffs move in the opposite direction. We show that under log-concave background risk, the willingness to pay for information decreases with risk aversion for investment menus. Under downward-log-convex background risk, it increases with risk aversion for insurance menus.

The mechanisms that explain these results are different for the two 
types of assets. With insurance menus, information can protect against adverse realizations of the background risk. When the support is unbounded below, and under downward log convexity, the fixed price has a relatively small effect on the lower-tail probabilities. But signal-contingent insurance can substantially improve outcomes in the worst states. For this reason, greater risk aversion will raise the value of information. On the other hand, with investment menus, information allows the decision maker to avoid exposure in the most adverse states and to select investment opportunities in more favorable ones. If the tail risk is log-concave, the downward wealth movement produced by the information cost has a stronger effect on lower-tail probabilities. This induces a lower willingness to pay when the risk aversion is higher.

In the proofs we first establish the existence of monotone optimal strategies under a MLR-ordering condition on the experiment. Then, we derive opposite conditional single-crossing comparisons for investment and insurance menus. The main difficulty in establishing this result is that these comparisons depend on the realization of background wealth and must be aggregated. For this reason, the main technical tool we use is a state-dependent single-crossing aggregation theorem. It extends the classical preservation result of \textcite{karlin1956}. We can apply it directly to investment menus. For insurance menus, we need to combine the same single-crossing principle on the lower region with a direct comparison in the rest of the distribution.

We also reach opposite conclusions to those in the main results by reversing tail behavior. For investment menus, strict log convexity reverses the log-concave geometry of \Cref{theorem:investment}, and willingness to pay for information may increase with risk aversion. For insurance menus, super-exponentiality is a counterpart to downward log convexity. While downward log convexity keeps $(\log f)'$ bounded above as wealth tends to $-\infty$, super-exponentiality makes it diverge to $+\infty$. This means that the information price can then dominate insurance in the extreme left tail. As a result, the willingness to pay may decrease with risk aversion.

The comparative statics determine which agents choose to become informed before making financial decisions. This selection may matter in broader models of financial and insurance markets, where the distribution of information across agents can affect price formation \parencite{GrossmanStiglitz1980,VanNieuwerburghVeldkamp2010} and adverse selection \parencite{RothschildStiglitz1976}. We do not analyze these equilibrium consequences here; our results instead identify a primitive force governing information acquisition that can be incorporated into such models. We also briefly consider mixed asset menus, comprising both investment and insurance opportunities, to examine how risk aversion affects the type of asset selected, separately from its effect on willingness to pay for information.

A broad literature in macroeconomics and finance emphasizes the economic importance of tail behavior \parencite{rietz1988,barro2006rare,gabaix2012rare}, while the literature on information and risk aversion obtains sign results that depend on the decision problem and payment convention \parencite{gould1974risk,hilton1981determinants,freixas1984risk,willinger1989risk,nadiminti1996risk,kihlstrom2022risk}. Our results connect these literatures by identifying asset cyclicality and the tail geometry of background risk as joint determinants of the sign of this comparative static within an expected-utility framework.

The remainder of the paper is organized as follows. \Cref{sec:motivating-examples} presents motivating examples. \Cref{sec:environment} introduces the model and defines willingness to pay for information. \Cref{sec:main-results} develops the main comparative-statics results as well as the aggregation arguments that support them. \Cref{sec:reversals} constructs reversals using the opposite tail geometry.
\Cref{sec:extensions-and-applications} first studies the impact of the tail assumptions. Then, it presents applications to financial options and continuous-state environments with multiple assets. Finally, it considers a different asset-selection question which is raised by mixed menus.
\Cref{sec:related-literature} discusses the related literature. \Cref{sec:conclusion} concludes. Proofs omitted from the text are collected in the \hyperref[app:start]{Appendix}.

\section{Motivating Examples}\label{sec:motivating-examples}

Before introducing the general model, we illustrate the two opposite comparative statics that motivate the paper. In both examples, information allows the decision maker to condition her choice of an asset on the realized state. This makes it valuable. The main difference between the two cases is the economic role of the asset. Information can be used to select a risky investment in the first case or an insurance opportunity in the second. The resulting comparative statics point in opposite directions. They thus suggest that the nature of the available assets is one important determinant of how the value of information varies with risk aversion.

For these examples we use on purpose two states and CARA utility. In this way, the investment-versus-insurance mechanism can be shown without a tail-shape condition. This is useful in turn to interpret the main results below. With only two support points, discrete log concavity or log convexity imposes no additional restriction. The log-concavity and downward-log-convexity assumptions in our main theorems are relevant because we then move to general background-risk distributions and a broad class of utility functions. There, our conditional distributional comparisons must be aggregated across wealth states.

\begin{example}
\label{example investment} 
Suppose that a businessperson with CARA utility $u_r(m)=-e^{-rm}$, $r>0$, is deciding
whether to invest in a technology startup producing an electric car. There are two states of the world: high (the economy goes generally well), and low (the economy tumbles). The payoff from not investing is $\left( k_{H},k_{L}\right) $, with $k_{H}>k_{L}$. The net payoff from investing is $b_{H}=H>0$ in the high state (electric mobility succeeds only in a strong economy), and $b_{L}=-L$ in the low state, where $L>0$. Thus, the investment is procyclical. We consider an information structure $\alpha $ that perfectly reveals the state.
Let $p\in(0,1)$ denote the prior probability of the low state, and assume that
\begin{equation*}
(1-p)H-pL<0.
\end{equation*}
so that a risk-neutral agent does not invest without information. Under the displayed inequality and $k_H>k_L$, the status quo is optimal for every $r>0$. Indifference between remaining uninformed and purchasing full information at price $I$ gives
$$
I=-\frac{1}{r}\log\left(\frac{(1-p)e^{-r(H+k_H)}+pe^{-rk_L}}{(1-p)e^{-rk_H}+pe^{-rk_L}}\right).
$$
Lemma~\ref{lemma:example1} shows that willingness to pay is strictly decreasing in $r$.

\begin{lemma}
\label{lemma:example1}
Let $0<p<1$, $H>0$, and $k_H>k_L$.
Define, for $r>0$, 
\begin{equation*}
g(r) =-\frac{1}{r}\ln\left( \frac{(1-p)e^{-r(H+k_H)} + p e^{-r k_L}} {%
(1-p)e^{-r k_H} + p e^{-r k_L}} \right).
\end{equation*}
Then $g$ is strictly decreasing on $(0,\infty)$.
\end{lemma}
\end{example}

\begin{example}
\label{example insurance} 
Suppose now that the businessperson has the same CARA utility and is deciding
whether to invest in a company that produces small prefabricated houses that are very cheap to build but are not very desirable. There are two states of the world: high (the economy goes generally well), and low (the economy tumbles). 

The payoff from not investing is $\left( k_{H},k_{L}\right) $, with $k_{H}>k_{L}$. The net payoff from investing is $b_{H}=-H<0$ in the high state (the cheap house does not sell well in a strong economy), and $b_{L}=L>0$ in the low state (where low-quality housing is in high demand), with $k_{L}+L<k_{H}-H$. Thus, the asset is countercyclical and acts as an insurance opportunity. We consider an information structure $\alpha $ that perfectly reveals the state.
Let $p\in(0,1)$ denote the prior probability of the low state, and assume that
\begin{equation*}
-\left( 1-p\right) H+pL<0,
\end{equation*}%
or equivalently%
\begin{equation*}
p<\frac{H}{H+L}
\end{equation*}%
so that a risk-neutral agent does not invest without information. Under the displayed restrictions, there is a unique $\bar r>0$ such that the status quo is optimal without information if and only if $0<r\leq\bar r$. For these coefficients, indifference between remaining uninformed and purchasing full information at price $I$ gives
$$
I=-\frac{1}{r}\log\left(\frac{(1-p)e^{-rk_H}+pe^{-r(L+k_L)}}{(1-p)e^{-rk_H}+pe^{-rk_L}}\right).
$$
Thus willingness to pay equals $g(r)$ in Lemma~\ref{lemma:example2}, which is strictly increasing in $r$.

\begin{lemma}
\label{lemma:example2}
Let $0<p<1$, $L>0$, and $k_H>k_L+L$.
Define, for $r>0$, 
\begin{equation*}
g(r) = -\frac{1}{r}\ln\left( \frac{(1-p)e^{-rk_H} + p e^{-r(L+k_L)}} {%
(1-p)e^{-rk_H} + p e^{-rk_L}} \right).
\end{equation*}
Then $g$ is strictly increasing on $(0,\infty)$.
\end{lemma}

\end{example}

The examples isolate the cyclicality channel. Information supports greater exposure in the good state in the investment example and greater protection in the bad state in the insurance example. Because the signal fully reveals the state, this channel alone determines the comparative statics among CARA coefficients for which the status quo is optimal without information. Proposition~\ref{prop:complete-information} in Appendix~\ref{sec:full revelation} establishes the corresponding distribution-free benchmark: under full revelation, willingness to pay decreases with risk aversion when every asset is an investment and increases when every asset is an insurance.

Imperfect information creates the second ingredient in the main results. Paying for information shifts wealth downward in every state, while noisy signal-contingent asset choice moves probability mass across a cutoff in selected states. The relevant conditional comparisons must therefore be aggregated across background-wealth realizations. Log concavity makes the price effect relatively important in the lower tail for investment menus. Downward log convexity gives the opposite local geometry for insurance menus, while an unbounded lower support rules out a hard-boundary effect for very risk-averse agents. Section~\ref{sec:environment} now introduces the general model, and Section~\ref{sec:assumption-diagnostic} returns to these distinct roles.

\section{Environment}\label{sec:environment}

This section introduces the decision problem and defines willingness to pay for information. The decision maker faces uncertain initial wealth, may purchase information about its realization, and then selects an asset whose payoff depends on that realization. We subsequently compare decision makers who face the same economic environment but differ in risk aversion.

The agent has uncertain wealth $W$ with CDF $F$, whose realization is denoted $w$. We refer to $F$ as the distribution of \emph{background risk}. We assume that $\int |w|\,\mathrm dF(w)<\infty$ and that $F$ admits a density $f:\mathbb R\to\mathbb R_+$.
She has a (strictly) increasing utility function over money $u : \mathbb{R} \to \mathbb{R}$, which is $F$-integrable, absolutely continuous on any bounded interval, and satisfies $\limsup_{m \to \infty} u(m)/m < \infty$.%
\footnote{A map $u : \mathbb{R} \to \mathbb{R}$ is $F$-integrable if it is measurable and such that $\int |u| \,\mathrm{d}F < \infty$; the growth condition on $u$ is used to ensure that all integrals involving $u$ that we consider are well defined.}
We let $U$ be the set of functions $\mathbb R\to\mathbb R$ satisfying these properties. Recall that, given $u,v \in U$, $v$ is more \emph{risk-averse} than $u$ if $v = \phi \circ u$ for some increasing and (weakly) concave $\phi : \mathbb{R} \to \mathbb{R}$.

An \emph{asset} describes an extra monetary payoff received by the agent as a function of the realization of $w$. It is given by an $F$-integrable map $b : \mathbb{R} \to \mathbb{R}$, where $b(w)$ is the amount paid by the asset when the background risk realization is $w$. The agent has access to a set of assets $B$. We assume that $B$ includes the status quo which gives $b(w) =0$ for every $w$, and which we denote $\mathbf{0}$. 
We also assume that $B$ is compact when endowed with the $L^1$-metric induced by $F$.%
\footnote{That is, we assume that any sequence $(b_n)_{n \in \mathbb{N}}$ in $B$ admits a subsequence $(b'_n)_{n \in \mathbb{N}}$ such that, for some $b \in B$, $\int |b'_n-b| \,\mathrm{d}F \to 0$ as $n \to \infty$.}

The agent may decide to purchase information prior to choosing an asset. In this case, she pays a price $\mu \ge 0$ and gains access to a Blackwell experiment:
\begin{equation*}
	\alpha = (X,(h_w)_{w \in \reals})
\end{equation*} 
where $X$ is a finite set of signals, $h_w$ is a probability mass function over $X$ for each $w \in \mathbb{R}$, and the maps $w \mapsto h_w(x)$ are measurable for each $x \in X$.
We assume that $\alpha$ is \emph{MLR-ordered}: $X\subset\mathbb R$, and the kernel
$(w,x)\mapsto h_w(x)$ is TP2, meaning that
$$
h_w(x)h_{w'}(x')
\ge
h_w(x')h_{w'}(x)
$$
whenever $w\le w'$ and $x\le x'$.%
\footnote{``TP2'' stands for ``totally positive of order 2''. See Section~4 of \textcite{jewitt1987} for a brief overview of the theory of total positivity.}

When all signal probabilities are
strictly positive, this condition is equivalent to requiring
$h_w(x')/h_w(x)$ to be nondecreasing in $w$ for every $x<x'$. We assume without loss that every signal has
positive ex ante probability; signals with zero ex ante probability are deleted from $X$.

For each $x\in X$, let $h(x)$ be the probability of observing signal $x$, let $F_x$ be the posterior distribution after observing $x$, and let $f_x$ be its density.%
\footnote{Thus, $h(x)=\int h_w(x)\,\mathrm{d}F(w)>0$ and
$f_x(w)=\frac{h_w(x)f(w)}{h(x)}.$}

We restrict attention to environments for which the status quo is optimal before information is acquired. We refer to this condition as zero asset under zero information (ZAZI):
\begin{equation}
	\label{eq:ZAZI}
	\tag{ZAZI}
	\mathbf{0} \in \argmax_{b \in B} \int u(w+b(w))\,\mathrm{d}F(w).
\end{equation}
The ZAZI condition is a clean way to isolate the value of using information to make the asset choice: without information, the agent does not acquire any asset.%
\footnote{Our hypotheses on $F$, $u$, and $B$ guarantee that the maximization problems in \eqref{eq:ZAZI} and \eqref{eq:info} below admit solutions. See \Cref{app:prop:mustar} for a proof.}

\subsection{Willingness to Pay for Information}

An asset-selection \emph{strategy} for the decision maker is a family $\beta = (\beta_x)_{x \in X} \in B^X$, describing the asset chosen after observing each signal realization. 
If the agent purchases information and uses strategy $\beta$, the CDF of her terminal wealth $G^\mu_\beta : \mathbb{R} \to [0,1]$ is given by 
\begin{equation}
		\label{eq:G_beta}
		G^\mu_\beta(m) = \sum_{x \in X} h(x) \int \mathbf{1}_{\{w + \beta_x(w)-\mu \le m\}} \,\mathrm{d}F_x(w).
\end{equation}
A strategy is \emph{optimal} for $u$ if it maximizes the expected payoff $\int u \,\mathrm{d}G^\mu_\beta$ across all strategies $\beta \in B^X$.%
\footnote{The objective may equal $-\infty$ for all $\beta \in B^X$. In this case, all strategies are optimal. Note that \eqref{eq:info} remains valid since the right-hand side is finite, by hypothesis.}
The agent (strictly) prefers to purchase information if 
\begin{equation}
	\label{eq:info}
	\max_{\beta \in B^X} \int u \,\mathrm{d}G^\mu_\beta  \ge \mathrel{(>)} \int u(w) 
    \,\mathrm{d}F(w).%
\end{equation}

The following proposition shows that the agent's purchasing decision is summarized by a unique willingness to pay for information.

\begin{proposition}
\label{prop:mustar} 
For any utility function $u \in U$, there exists a unique price $M(u) \in \mathbb{R}_+$ such that the agent strictly prefers not purchasing at prices $\mu>M(u)$,
and she strictly prefers purchasing at prices $\mu<M(u)$.
\end{proposition}

The proof of Proposition \ref{prop:mustar} is presented in \Cref{app:prop:mustar}. 

\subsection{Comparative Statics in Risk Aversion}

We are interested in establishing conditions on the economic environment, given by background risk and available assets, such that more risk-averse agents are willing to pay more (resp.\ less) than less risk averse ones for information, prior to choosing an asset. 

\begin{enumerate}[label=(\Roman*)]
\item We say that \emph{willingness to pay for information decreases with risk aversion} if, for every $u,v\in U$ such that $u$ is more risk averse than $v$ and \eqref{eq:ZAZI} holds for both,
$$
M(u)\le M(v).
$$

\item We say that \emph{willingness to pay for information increases with risk aversion} if, for every $u,v\in U$ such that $u$ is more risk averse than $v$ and \eqref{eq:ZAZI} holds for both,
$$
M(u)\ge M(v).
$$
\end{enumerate}

The next section identifies sufficient conditions for the two orderings: investment menus under log-concave background risk yield decreasing willingness to pay, whereas insurance menus under downward-log-convex background risk yield increasing willingness to pay.

\section{Main Comparative-Statics Results}\label{sec:main-results}

This section establishes the paper's two positive comparative-statics results and develops their common distributional foundation. We first study investment menus under log-concave background risk and then insurance menus under downward-log-convex background risk. The proof strategy has four steps: translate each comparative static into a single-crossing criterion, select a monotone optimal strategy, derive the conditional wealth comparison induced by that strategy, and aggregate the conditional comparisons over background risk. Finally, we present a new aggregation result, of interest in its own right, for single-crossing functions that provides the main technical ingredient of the analysis before completing the proofs of the main theorems.

\subsection{Investment Menus under Log-Concave Background Risk}
\label{sec:investment-menus}

We first study environments in which the available assets increase exposure to background risk. Such assets are naturally interpreted as \emph{procyclical}: they perform better in states in which background wealth is already high, and therefore amplify fluctuations in total wealth.

An \emph{investment} is an asset whose payoff is nondecreasing in background wealth. Our results concern menus of investment opportunities rather than isolated assets. The following definition identifies menus whose elements can be ordered according to the amount of cyclical exposure they provide.

\begin{definition}
\label{def:investment}
$B$ is an \emph{investment menu} if each $b \in B$ is an investment and there exists $w^* \in \mathbb{R}\cup\{\pm\infty\}$ and a complete ordering $\succeq$ of $B$ such that, given any $b,b'\in B$ with $b'\succeq b$,
$$
b'(w)\le b(w)\qquad\text{for all }w\le w^*,
$$
and
$$
b'(w)\ge b(w)\qquad\text{for all }w\ge w^*.
$$
\end{definition}

Thus, moving upward in the menu shifts payoffs from relatively poor states toward relatively rich states. Higher-ranked assets therefore provide greater exposure to aggregate conditions while preserving a common ordering across all states. Above the crossing point $w^*$ higher-ranked assets pay more, while they pay less below it. 

Investment menus arise naturally in several settings. One example is
$$
B=\{\lambda b:\lambda\in L\},
$$
where $b$ is a nondecreasing $F$-integrable payoff, $L\subset\mathbb R_+$ is compact and contains $0$, and there exists $w^*\in\mathbb R\cup\{\pm\infty\}$ such that $b(w)\leq0$ for $w\leq w^*$ and $b(w)\geq0$ for $w\geq w^*$. More generally, if $b^1,\ldots,b^n$ are investments with a common crossing point $w^*$, then
$$
\{\mathbf{0},b^1,b^1+b^2,\ldots,b^1+\cdots+b^n\}
$$
is an investment menu, also with crossing point $w^*$. Additional examples are given in \Cref{sec:financial-options,example:continuous-electric-car}.

For investment menus, the relevant property of background risk is log concavity. Throughout this subsection we assume that the background-risk density is log concave, that is, $\{w \in \mathbb{R}: f(w) > 0\}$ is an interval on which $\log f$ is concave. 
Many standard distributions satisfy this condition, including the Gaussian distribution.

Our first main theorem shows that these two ingredients together imply a clean monotone comparative statics result.

\begin{theorem}
\label{theorem:investment}
If background risk is log concave and $B$ is an investment menu, then willingness to pay for information is decreasing with risk aversion.
\end{theorem}

The intuition begins with the benchmark of deterministic initial wealth studied by \textcite{CabralesGossnerSerranoEntropyAER2013,cabrales2017normalized}. 
Investment opportunities increase exposure precisely in states in which wealth is already high. Information becomes valuable because it allows the decision maker to select when to take this exposure. Less risk-averse decision makers are precisely those more willing to bear the resulting risk and thus value this information more highly.

Background risk brings in an additional screening motive. Information allows the decision maker to avoid investing when the state is catastrophic while preserving investment opportunities in nearby, less adverse states. Thus, for a very risk-averse agent, the value of information depends on how well it separates these two parts of the lower tail. Purchasing information also shifts wealth downward by a fixed amount in every state.

The thickness of tail distributions determines which of these effects dominates. When the tails are thin, as it happens with log concavity, the fixed downward movement of wealth has a strong effect on lower-tail risk, relative to the likelihood-ratio reweighting generated by the signal. The price of information is thus important for agents who are most affected by the worst outcomes. Thus, under log concavity, the screening benefit does not overturn the usual investment logic: less risk-averse agents remain willing to pay more for information.

\subsection{Insurance Menus under Downward-Log-Convex Background Risk}
\label{sec:insurance-menus}

We now consider environments in which the available assets reduce exposure to background risk.

An \emph{insurance} is an asset whose payoff is nonincreasing in background wealth and such that $w+b(w)$ is nondecreasing in $w$. Such assets transfer wealth from favorable states toward unfavorable ones, thereby reducing the cyclical component of wealth.

We extend this notion from individual assets to menus.

\begin{definition}
\label{def:insurance}
$B$ is an \emph{insurance menu} if each $b\in B$ is an insurance and there exists
$w^*\in\mathbb{R}\cup\{\pm\infty\}$ and a complete ordering $\succeq$ of $B$
such that, given any $b,b'\in B$ with $b'\succeq b$, $b'(w)\ge b(w)$ for all
$w\le w^*$ and $b'(w)\le b(w)$ for all $w\ge w^*$.
\end{definition}

One example of an insurance menu is
$$
B=\{\lambda b:\lambda\in L\},
$$
where $b:\mathbb R\to\mathbb R$ is nonincreasing, $F$-integrable, satisfies $\inf b<0<\sup b$, and $w+b(w)$ is nondecreasing, while $L$ is a compact subset of $[0,1]$ containing $0$.

For insurance menus, the relevant property of background risk is a version of log convexity.

\begin{definition}
A density $f$ is \emph{downward log convex} if
$S=\{w\in\mathbb R:f(w)>0\}$ is convex and unbounded
below, and $\log f$ is convex on $S$.%
\footnote{Note that, if $f$ is downward log convex, then it has support $(-\infty,\omega]$ for some $\omega \in \mathbb{R}$. The requirement that the positive-density region be unbounded below is separate from log convexity. With support bounded below, sufficiently risk-averse agents may not acquire information.}
\end{definition}

Our second main result mirrors \Cref{theorem:investment}.

\begin{theorem}
\label{theorem:insurance}
If background risk is downward log convex and $B$ is an insurance menu, then willingness to pay for information is increasing with risk aversion.
\end{theorem}

The economic force is the mirror image of the investment case. Insurance transfers wealth toward adverse realizations of background risk. In this case, information is valuable because it allows protection to be directed toward the most unfavorable states. A very risk-averse decision maker gives a high value to this protection.

But purchasing information also reduces wealth by a fixed amount in every state. The value of information weighs the gain from directing insurance toward catastrophic outcomes against the deterioration in tail risk that this uniform wealth loss causes.

The thickness of the left tail of the distribution is precisely what determines which of these effects dominates. Under fat left tails, specifically under downward log convexity, the fixed translation of wealth has a relatively weak effect on the lower-tail probabilities. On the other hand, the likelihood-ratio updating can change a lot the relative weight placed on nearby adverse states. In this way, the improvement of the targeting of insurance achieved by information can outweigh the worsening downside risk entailed by the price. Since the more risk-averse decision makers attach greater value to this reallocation of protection toward the worst states, the willingness to pay for information increases with risk aversion.

The next subsection develops the common distributional logic behind the two results.

\subsection{A Distributional Approach}
\label{sec:proof-roadmap}

Although the economic mechanisms differ, both comparative-statics results are ultimately driven by the same distributional comparison. Purchasing information brings together two forces with opposing impacts. One is a uniform downward wealth shift generated by its price. The other is a reallocation of wealth through a signal-contingent asset choice. For investments, the reallocation moves exposure away from catastrophic states. For insurance, it moves protection toward them. The proofs of \Cref{theorem:investment,theorem:insurance} thus proceed by comparing the distributions of terminal wealth with and without information.

The proof is built in four steps. First, Lemma~\ref{lemma:ptp_char} reduces each comparative static to a single-crossing comparison between terminal wealth with and without information. Second, Lemma~\ref{lemma:mon} exhibits a monotone optimal strategy in the relevant menu order. Third, Lemmas~\ref{lemma:investment-conditional} and~\ref{lemma:insurance-conditional} use Lemma~\ref{lemma:tp2-fosd} to derive the appropriate conditional comparisons for investment and insurance menus. Finally, Lemma~\ref{lemma:sc_random} applies the state-dependent preservation result in Lemma~\ref{lemma:sc_incr} to aggregate those comparisons over background risk. We collect the auxiliary proofs in Appendix~\ref{app:proof-roadmap} and Appendix~\ref{app:main-proofs}. Appendix~\ref{sec:worked examples} gives worked illustrations of the monotone strategies and terminal-wealth crossings.

Recall from \Cref{sec:environment} that $G^\mu_\beta$ is the terminal-wealth distribution after purchasing information at price $\mu \ge 0$ and implementing strategy $\beta \in B^X$.
For a non-empty $I \subseteq \mathbb{R}$, a map $\psi:I \to\mathbb R$ is \emph{single crossing} if, for all $m\le m'$,
$$
\psi(m)>0
\quad\Longrightarrow\quad
\psi(m')\ge0.
$$

\begin{lemma}[Single-crossing criterion]
\label{lemma:ptp_char}
Willingness to pay for information decreases with risk aversion if, for every $u\in U$ satisfying \eqref{eq:ZAZI} and every $0\le\mu<M(u)$, there exists a strategy $\beta$ that is optimal for $u$ at price $\mu$ and such that $F-G^\mu_\beta$ is single crossing.

Willingness to pay for information increases with risk aversion if, for every $u\in U$ satisfying \eqref{eq:ZAZI} and every $0\le\mu\le M(u)$, there exists a strategy $\beta$ that is optimal for $u$ at price $\mu$ and such that $G^\mu_\beta-F$ is single crossing.
\end{lemma}

Thus, the economic question is transformed into a geometric one: it is enough to show that $F-G^\mu_\beta$ in the decreasing case, or $G^\mu_\beta-F$ in the increasing case, is single crossing. Once the relevant CDF difference becomes strictly positive, it remains nonnegative at every higher cutoff.
This lemma is a direct consequence of Theorem~1 in \textcite{jewitt1989}.

The second lemma shows that optimal strategies inherit the ordering structure of the asset menu.
Say that a strategy $\beta\in B^X$ is \emph{monotone} relative to $\succeq$ if
$x\le y$
implies $\beta_y\succeq\beta_x$.

\begin{lemma}[Monotone optimal strategy]
\label{lemma:mon}
Suppose that $\succeq$ is a complete ordering of $B$ such that $b'-b$ is single crossing whenever $b'\succeq b$. Then, for every $u\in U$ and every $\mu\ge0$, there exists a strategy that is optimal for $u$ at price $\mu$ and monotone relative to $\succeq$.
\end{lemma}

For an investment menu, the lemma shows that higher signals lead to weakly higher-ranked investments. For an insurance menu, it shows that higher signals lead to weakly lower-ranked insurance positions.

It remains to understand how these monotone conditional decisions aggregate once background risk is integrated out.

In what follows, given random variables $Y$ and $W$, we denote by $F_Y$ the distribution function of $Y$, and by $F_{Y|W}$ a conditional distribution of $Y$ given $W$.%
\footnote{As usual, $F_{Y|W}$ denotes any measurable version of the conditional distribution satisfying
$$
\int_E F_{Y|W}(y|w)\,\mathrm{d}F_W(w)
=
\Pr(Y\le y,\;W\in E)
$$
for every Borel set $E\subseteq\mathbb R$.}

The first two lemmas reduce the proofs of the main theorems to a distributional comparison between wealth with and without information. The following lemma is used to provide this precise comparison. It shows that the single-crossing structure induced by monotone asset choices is preserved after adding background risk, under an appropriate tail condition on background risk. This probabilistic step is the one connecting the economic structure of investment and insurance menus with the distributional criterion of \Cref{lemma:ptp_char}. In its proof we use a direct application of the aggregation theorem established in the next subsection.

\begin{lemma}
\label{lemma:sc_random}
Let $(Y^1,Y^2,W)$ be a random vector. Suppose that versions of $F_{Y^1|W}$ and $F_{Y^2|W}$ can be chosen such that there exist $\underline w\le\bar w$ in $\mathbb R\cup\{\pm\infty\}$ for which, for every $y\in\mathbb R$, the map $w\mapsto
F_{Y^1|W}(y|w-y)-F_{Y^2|W}(y|w-y)$ is nonpositive on $\{w \in \mathbb{R} : w \le \underline w\}$, nondecreasing on
$(\underline w,\bar w)$, and nonnegative on $\{w \in \mathbb{R}: w \ge \bar w\}$.

Then, the following two statements hold:
\begin{enumerate}[label=(\Roman*)]
    \item \label{item:sc_random:logconcave} If $W$ admits a log-concave density and, for every $w\in\mathbb R$,
    $y\mapsto F_{Y^1|W}(y|w-y)-F_{Y^2|W}(y|w-y)$   is single crossing, then
    $F_{Y^1+W}-F_{Y^2+W}$
    is single crossing.
    \item \label{item:sc_random:logconvex} If $W$ admits a downward-log-convex density with support $(-\infty,\omega]$, 
    $y\mapsto F_{Y^2|W}(y|w-y)-F_{Y^1|W}(y|w-y)$
    is single crossing for each $w \in \mathbb{R}$, and $F_{Y^1|W}(y|w-y) = F_{Y^2|W}(y|w-y)$ for all $y < w-\omega$ and $w < \bar w$, then 
    $F_{Y^1+W}-F_{Y^2+W}$ is single crossing.
\end{enumerate}
\end{lemma}

\subsection{Aggregation under Background Risk}
\label{sec:single-crossing-aggregation}

The final step is to aggregate the conditional wealth comparisons over background risk. Standard preservation results, such as the ones found in \textcite{karlin1956}, cannot be applied directly because the conditional comparison here depends on background wealth. In our case both the conditional signal probabilities and the state-contingent asset payoffs vary with the state. The following result, proved in Appendix~\ref{app:proof-roadmap}, shows that one can preserve single crossing when this state dependence is monotone and the aggregation kernel is TP2. In this way, it provides the common mathematical foundation for both comparative-statics theorems.

\begin{lemma}
\label{lemma:sc_incr}
Let $I\subseteq\mathbb R$ be an interval, let
$
\phi:\mathbb R\times I\to\mathbb R,
\kappa:\mathbb R\times I\to\mathbb R_+,
$
and suppose that $\kappa$ is TP2. Assume that:

\begin{enumerate}[label=(\roman*)]
\item for every $m\in I$, the map
$\ell\longmapsto\phi(\ell,m)$ is single crossing;

\item there exist
$-\infty\le\underline m\le\bar m\le\infty$
such that, for every $\ell\in\mathbb R$, the map
$m\longmapsto\phi(\ell,m)$ is nonpositive on $I\cap(-\infty,\underline m]$, nondecreasing on
$I\cap(\underline m,\bar m)$, and nonnegative on
$I\cap[\bar m,\infty)$;

\item 
for every $\ell$ and $m$ in $I$, the map
$k\longmapsto\phi(k,\ell)\kappa(k,m)$ is integrable. 
\end{enumerate}

Then
$m\longmapsto
\int_{\mathbb R}\phi(\ell,m)\kappa(\ell,m)\,d\ell$
is single crossing on $I$.
\end{lemma}

The novelty relative to the classical preservation result of \textcite{karlin1956} (\Cref{lemma:sc} in \Cref{app:proof-roadmap}) is that the conditional comparison is allowed to vary with the aggregation variable.
In our application, there is variation, because the conditional distribution of the selected payoff, and hence the conditional wealth comparison, varies with the state.
The monotonicity condition in the second argument makes sure that there is a coherence in the movement of comparisons with background wealth. At the same time, TP2 ensures that integration preserves their single-crossing structure. The classical result of \textcite{karlin1956} can be recovered by considering the case $I=\mathbb R$ and $\phi$ independent of its second argument.

\subsection{Proofs of the Main Theorems}
\label{sec:main-proofs}

We now combine the preceding results. Lemma~\ref{lemma:ptp_char} reduces each comparative-statics result to a single-crossing comparison between terminal wealth with and without information, while Lemma~\ref{lemma:mon} provides an optimal monotone strategy. For investment menus, \Cref{lemma:sc_random}\ref{item:sc_random:logconcave} aggregates the resulting conditional comparison under log-concave background risk, while \Cref{lemma:sc_random}\ref{item:sc_random:logconvex} performs the same task for insurance menus under downward log convex background risk. 

The remaining conditional properties of monotone strategies, necessary to apply Lemma~\ref{lemma:sc_random}, are established in \Cref{lemma:investment-conditional,lemma:insurance-conditional} in \Cref{app:main-proofs}.

\begin{proof}[Proof of \Cref{theorem:investment}]
Let $w^*$ and $\succeq$ satisfy \Cref{def:investment}. If
$w^*\in\{\pm\infty\}$, then either \eqref{eq:ZAZI} fails for every nontrivial utility under consideration or every admissible asset is weakly
dominated by the status quo, so the conclusion is immediate. Hence assume $w^*\in\mathbb R$.

Fix $u\in U$ and $\mu\ge0$. Let $\chi$ denote the signal. By \Cref{lemma:mon}, applied to $\succeq$, there exists an optimal strategy $\beta$ satisfying $\beta_y\succeq\beta_x$ whenever $x\le y$, hence monotone relative to $\succeq$. By \Cref{lemma:investment-conditional}, the conditional distributions
associated with
$Y^1=0$,
$Y^2=\beta_\chi(W)-\mu$
satisfy the hypotheses of \Cref{lemma:sc_random}\ref{item:sc_random:logconcave}. Hence $F-G^\mu_\beta$ is single crossing. The conclusion follows from \Cref{lemma:ptp_char}.
\end{proof}

\begin{proof}[Proof of \Cref{theorem:insurance}]
Let $w^*$ and $\succeq$ satisfy \Cref{def:insurance}, and let $f$ be downward log convex with support $(-\infty,\omega]$. If
$w^*\in\{-\infty\}\cup[\omega,\infty]$, comparison with $\mathbf{0}$ makes every asset one-signed on the support of $F$. The ZAZI condition excludes any nonnegative nonzero asset, while every nonpositive asset is dominated by the status quo. The conclusion is therefore immediate. Hence assume $w^*\in(-\infty,\omega)$.

Let $\succeq^R$ be the reverse of $\succeq$: $b'\succeq^R b$ if and only if $b\succeq b'$. If $b'\succeq^R b$, then $b'-b$ is single crossing. Fix $u\in U$ and $\mu\ge0$. By \Cref{lemma:mon}, applied to $\succeq^R$, there exists an optimal strategy $\beta$ satisfying $\beta_y\succeq^R\beta_x$, equivalently $\beta_x\succeq\beta_y$, whenever $x\le y$. Let $\chi$ denote the signal and define
$
Y^1=\beta_\chi(W)-\mu$,
$Y^2=0.$
By \Cref{lemma:insurance-conditional}, the associated conditional distributions satisfy the hypotheses of \Cref{lemma:sc_random}\ref{item:sc_random:logconvex}. Hence $G^\mu_\beta-F$ is single crossing. The conclusion follows from \Cref{lemma:ptp_char}.
\end{proof}

\section{Comparative-Statics Reversals}\label{sec:reversals}

Theorems~\ref{theorem:investment} and \ref{theorem:insurance} identify two broad classes of background-risk distributions under which the value of information is monotone in risk aversion. We now show that their tail restrictions have substantive force by constructing environments in which each comparative statics result reverses. For investment menus, the reversal arises under strict log convexity. For insurance menus, it arises under a super-exponential left tail, whose limiting behavior is opposite to downward log convexity. These constructions complement the main results without asserting that their assumptions are necessary.

\subsection{Investment Menus under Log-Convex Tails}\label{sec:investment-reversals}

We first consider investment menus under log-convex background risk. In contrast with \Cref{theorem:investment}, willingness to pay for information need no longer decrease with risk aversion. The following result is proved in \Cref{app:invest-log-convex}.

\begin{proposition}[Investment Reversal under Log-Convex Tails]
\label{prop:invest-log-convex}
Suppose that $f$ is continuously differentiable and strictly positive on $(-\infty,\omega)$ for some $\omega \in \mathbb{R}$, and that $\log f$ is strictly convex on this interval. 

Then there exist an investment menu $B=\{\mathbf{0},b\}$, a binary MLR-ordered experiment $\alpha$, and increasing, Lipschitz continuous, concave utility functions $u$, $v$ satisfying \eqref{eq:ZAZI} such that $u$ is more risk
averse than $v$ and
$$
M(u)>M(v).
$$
\end{proposition}

The construction is based on an investment $b$ that generates a large loss in an adverse wealth region and a gain otherwise. Without information, exposure to the loss makes the investment unattractive. The experiment $\alpha$ sometimes generates a positive signal in the favorable region, but never in the adverse region. The positive signal thus induces an investment decision. 

Consider the calibrated price $\mu$ at which the benchmark decision maker $v$ is indifferent between purchasing information and remaining uninformed. Also, let $F$ and $G$ denote the corresponding terminal-wealth distributions. In the extreme left tail, the favorable signal never occurs, so only the information price matters and $G>F$. The calibration yields $G<F$ on an intermediate lower region. Then, strict log convexity yields $G>F$ on a higher interval. This means that the CDF difference changes sign at least twice. We calibrate the favorable intermediate region to dominate the unfavorable extreme-tail in the integrated comparison. Greater risk aversion increases the relative weight placed on this net lower-region improvement.

\subsection{Insurance Menus under Super-Exponential Left Tails}\label{sec:insurance-reversals}

We next consider the opposite tail behavior to that implied by downward log convexity. We say that $f$ has a \emph{super-exponential left tail} if it is continuously differentiable and strictly positive on $(-\infty,\omega)$ for some $\omega\in\mathbb R$ and
$$
\lim_{w\to-\infty}(\log f)'(w)=+\infty.
$$
Under downward log convexity, $(\log f)'$ is nondecreasing and therefore remains bounded above as wealth tends to $-\infty$. Super-exponentiality imposes the opposite limiting behavior: the log slope becomes arbitrarily large in the extreme left tail. In particular, the left tail is thinner than that of any exponential distribution. The following proposition, proved in \Cref{app:insurance-super-exponential}, shows that under this tail behavior the insurance comparative statics of \Cref{theorem:insurance} may reverse.

\begin{proposition}[Insurance Reversal under Super-Exponential Left Tails]
\label{prop:insurance-super-exponential}
Suppose that $f$ has a super-exponential left tail. Then there exist an insurance menu $B=\{\mathbf{0},b\}$, a binary MLR-ordered experiment $\alpha$, and increasing, Lipschitz continuous, concave utility functions $u,v$ satisfying \eqref{eq:ZAZI} such that $u$ is more risk averse than $v$ and
$$
M(u)<M(v).
$$
\end{proposition}

The benchmark decision maker $v$ is risk neutral. Our construction relies on an insurance contract that pays a fixed amount in sufficiently adverse states and produces losses in more favorable ones. After a low signal, the contract pays the fixed amount throughout the posterior support. For this reason it is strictly preferred by every increasing decision maker. Without information, the benchmark decision maker rejects the contract because its prior expected payoff is negative. She also rejects it after a high signal, which downweights the adverse states in which the contract pays. The information price is set exactly equal to the expected gross payoff from this signal-contingent strategy. This makes $v$ indifferent between purchasing information and remaining uninformed.

Let $F$ denote the terminal-wealth distribution when information is not purchased by $v$ and $G$ when she does and follows this strategy. In sufficiently adverse states, the low signal occurs with probability $q\in(0,1)$. After a high signal, no insurance is selected, but the information price is still paid. For a sufficiently low cutoff $m$, this high signal contributes $(1-q)F(m+\mu)$ to $G(m)$.

Super-exponentiality implies that $F(m+\mu)/F(m)\to+\infty$ as $m\to-\infty$. This means that the branch where price is the only effect eventually makes $G(m)>F(m)$, in spite of the protection provided after the low signal. Since $v$ is linear and indifferent, $F$ and $G$ have the same mean. We obtain the utility $u$  by increasing the slope of $v$ only in this extreme lower-tail region. This makes $u$ more risk averse and gives greater weight to worst outcomes. By choosing the change sufficiently small, the prior and posterior insurance choices made by $v$ remain optimal for $u$, while information becomes strictly unattractive at the benchmark price.

Together, Propositions~\ref{prop:invest-log-convex} and~\ref{prop:insurance-super-exponential} demonstrate that the tail restrictions in the positive theorems have a substantive effect. Alternative tail geometries can reverse each comparative statics result. Section~\ref{sec:assumption-diagnostic} isolates the local probability-mass comparison behind the positive and reversal results.


\section{Extensions and Applications}\label{sec:extensions-and-applications}

This section first isolates the roles of tail curvature and support in the main theorems. It then applies the results to financial options and to continuous-state environments with multiple investment or insurance opportunities. The final subsection considers a related but distinct question: conditional on acquiring information, how does risk aversion affect the choice between an investment and an insurance opportunity?

\subsection{The Role of the Tail Assumptions}\label{sec:assumption-diagnostic}
The two-state examples in \Cref{sec:motivating-examples} were designed to isolate asset cyclicality from the geometry of background risk.
The information is used either to take larger exposure selectively or to direct protection toward bad states. In those examples, no shape restriction on background risk is needed. In the general model, the assumptions on the tail become relevant. This is so because purchasing information and using it in the asset choice move wealth in opposite directions. We first show a simple local calculation that neatly captures the geometry that underlies the more general aggregation arguments.

Fix a wealth cutoff $m$ and two positive shifts $a$ and $\mu$. Consider the adjacent probability masses $A^+(m)=F(m+\mu)-F(m)$ and $A^-(m)=F(m)-F(m-a)$. Whenever $A^+(m)+A^-(m)>0$, define their relative share by
$$
P_{a,\mu}(m)=\frac{A^+(m)}{A^+(m)+A^-(m)}=\frac{F(m+\mu)-F(m)}{F(m+\mu)-F(m-a)}.
$$
The numerator is the amount of mass just above $m$ which a uniform price $\mu$ can push below the cutoff. The term $A^-(m)$, on the other hand, is the mass just below $m$ which may be moved above it by a state-contingent upward payoff of size $a$. Thus $P_{a,\mu}(m)$ becomes a good local measure of the relative strengths of the price effect versus the asset-reallocation effect. The full proofs have to take into account in addition the endogenous asset choice and the signal structure. However, it is this adjacent-mass comparison in the tail geometry that makes the aggregation work.

The curvature of $\log f$ determines how this comparison varies with the cutoff under the hypotheses of Lemma~\ref{lemma:left-tail-geometry} when the two adjacent intervals lie in the positive-density region. Under strict log concavity, $P_{a,\mu}(m)$ decreases with $m$, so the price effect becomes relatively more important as the cutoff moves farther into the left tail. Under strict log convexity, $P_{a,\mu}(m)$ increases with $m$, so the asset-reallocation effect becomes relatively more important farther into the left tail.

These are the local geometries behind the investment and insurance results, respectively.

These monotonicity statements are formalized in \Cref{lemma:left-tail-geometry}. The two reversal results use different tail comparisons. Strict log convexity reverses the adjacent-mass ordering behind the investment result. The insurance reversal instead exploits super-exponentiality: for every $\mu>0$, $F(m+\mu)/F(m)\to+\infty$ as $m\to-\infty$, and consequently $P_{a,\mu}(m)\to1$. Thus a branch that bears the information price without receiving insurance eventually dominates the baseline lower-tail probability. A log-linear density, such as $f(w)=e^w$ on $(-\infty,0]$, lies at the boundary of the adjacent-mass comparison because the corresponding ratio is constant.

In addition, we use the support assumption in the insurance theorem in a different way from log convexity. Curvature serves to control the relative mass of neighboring regions away from a boundary. The requirement that the support is unbounded below ensures that this comparison stays relevant as far as necessary into the bad tail. If the lower endpoint were finite, a sufficiently risk-averse agent's preferences could be dominated by what happens near that boundary. She may then cease to purchase costly information. In this case an increasing comparative statics result cannot be guaranteed. This is why our definition of downward log convexity combines convexity of $\log f$ with support extending to $-\infty$. The reflected Lomax distribution used in \Cref{example:continuous-prefabricated-house} satisfies both requirements.

We then see that the assumptions have distinct roles. The investment-versus-insurance menu condition determines the direction in which information reallocates wealth across states. The assumption on tail curvature determines how that reallocation compares with the universal wealth loss from the information price.

The next two subsections give examples where these proof ingredients are used in economically relevant situations.

\subsection{Financial Options}\label{sec:financial-options}

The sufficient conditions in \Cref{theorem:investment,theorem:insurance} apply to option payoffs when the underlying payoff is a monotone function of background wealth. Let $S=s(w)$, where $s$ is nondecreasing, and let $\pi\geq0$. A European call with strike $K$ and premium $\pi$ has net payoff $b_C(w)=(s(w)-K)^+-\pi$. Assume that $b_C$ is $F$-integrable and has a zero $w^*$. Then, for any compact $L\subseteq\mathbb R_+$ containing $0$, $B_C=\{\lambda b_C:\lambda\in L\}$ is an investment menu: if $\lambda'\geq\lambda$, then $\lambda'b_C\leq\lambda b_C$ below $w^*$ and $\lambda'b_C\geq\lambda b_C$ above $w^*$. Under log concavity of the background-risk density, \Cref{theorem:investment} implies that, among decision makers satisfying ZAZI, willingness to pay for information before choosing the call position decreases with risk aversion.

Put options provide the corresponding insurance example. Consider a put written directly on wealth, with strike $K$, premium $\pi\geq0$, and net payoff $b_P(w)=(K-w)^+-\pi$. This payoff is nonincreasing and is $F$-integrable under the maintained first-moment assumption. For every $\lambda\in[0,1]$, the map $w+\lambda b_P(w)$ is nondecreasing: its slope is $1-\lambda$ below $K$ and $1$ above $K$. Since $b_P$ crosses zero at $w^*=K-\pi$, $B_P=\{\lambda b_P:\lambda\in[0,1]\}$ is an insurance menu. Under downward log convexity of the background-risk density, \Cref{theorem:insurance} implies that, among decision makers satisfying ZAZI, willingness to pay for information before choosing the put position increases with risk aversion.

These examples also illustrate the importance of the menu condition. A single call option in which the position size is variable satisfies the condition about investment-menus. A single put option where the position is limited, satisfies the insurance-menu condition. On the other hand, if we use an arbitrary collection of options where strikes and premia are different, we need not satisfy either condition. This is so because different options are able to cross the status quo at different wealth levels. This means that the theorems can be applied in the most direct way to derivative menus ordered by a common exposure parameter.

\subsection{Continuous-State Examples with Multiple Assets}

We now provide continuous-state analogues of the electric-car and prefabricated-house examples in \Cref{sec:motivating-examples}. The diagnostic in \Cref{sec:assumption-diagnostic} explained why distributional shape matters; the purpose of the present examples is complementary. We want to show that the menu,
tail-shape, and support assumptions of the main theorems can arise together in economically natural environments with several assets. The state variable $w$ should be interpreted as the realization of the general economic environment, or equivalently as the businessperson's background wealth. Higher values of $w$ correspond to stronger aggregate conditions.

Let the information structure be binary, $X=\{L,H\}$. For some $\gamma>0$ and threshold $w^{\ast}$, suppose
$$
\begin{aligned}
h_w(H)&=\frac{\exp(\gamma(w-w^{\ast}))}{1+\exp(\gamma(w-w^{\ast}))},&
h_w(L)&=\frac{1}{1+\exp(\gamma(w-w^{\ast}))}.
\end{aligned}
$$
Then $h_w(H)/h_w(L)=\exp(\gamma(w-w^{\ast}))$, which is increasing in $w$. Hence the experiment is MLR-ordered. The signal $H$ is good news about the aggregate state, while $L$ is bad news.

\begin{example}[A portfolio of electric-car investments]
\label{example:continuous-electric-car} 
Suppose that a businessperson with CARA preferences, $u_{r}(m)=-\exp (-rm)$, $r>0$, is deciding how much to invest in a collection of complementary projects related to electric mobility: for instance, an electric-car platform, a battery facility, charging infrastructure, and software. These projects are cyclical: they perform well when the aggregate economy is strong and perform poorly when the aggregate economy is weak.

Assume that background wealth $W$ has a log-concave density. For example, let $W\sim N(m,\sigma ^{2})$. For each project $i=1,\ldots ,n$, let $e_{i}: \mathbb{R}\rightarrow \mathbb{R}$ denote the net payoff from project $i$.

Assume that each $e_{i}$ is nondecreasing, $F$-integrable, and has the same crossing point $w^{\ast }$: $e_{i}(w)\leq 0$ for $w\leq w^{\ast }$,  $e_{i}(w)\geq 0$ \quad \text{for } $w\geq w^{\ast }$. A simple parametric specification is $e_{i}(w)=a_{i}(w-w^{\ast })$, $a_{i}>0$. Let the menu consist of cumulative investment plans, $b_{E}^{k}(w)=\sum_{i=1}^{k}e_{i}(w),\qquad k=1,\ldots ,n,$ together with the status quo: $B_{E}=\{\mathbf{0},b_{E}^{1},\ldots ,b_{E}^{n}\}$. Thus $b_{E}^{1}$ corresponds to investing only in the core electric-car platform, $b_{E}^{2}$
adds a second complementary project, and so on.

By the cumulative-menu construction following \Cref{def:investment}, the monotonicity and common-crossing assumptions on the $e_i$ imply that $B_E$ is an investment menu, ordered by $b_E^{k'}\succeq b_E^k$ if and only if $k'\geq k$.

In the Gaussian linear specification, with $W\sim N(m,\sigma^2)$ and $e_i(w)=a_i(w-w^*)$, write $A_k=\sum_{i=1}^k a_i$. Then $w+b_{E}^{k}(w)=(1+A_{k})w-A_{k}w^{\ast }$. If $w^{\ast }\geq m$, the status quo is optimal under the prior for every CARA coefficient $r>0$. Indeed, 
\begin{equation*}
\mathbb{E}\left[\exp{-r(W+b_E^k(W))}\right] = \exp\left( -r\bigl(%
(1+A_k)m-A_k w^*\bigr) +\frac{r^2}{2}(1+A_k)^2\sigma^2 \right),
\end{equation*} 
which is increasing in $A_{k}\geq 0$ whenever $w^{\ast }\geq m $. Since CARA utility is the negative of this expression, expected utility is maximized at $A_{k}=0$, i.e. at the status quo.

The conclusion of \Cref{theorem:investment} applies to this example. Thus, if $r'>r>0$, then $M(u_{r'})\leq M(u_r)$.
\end{example}

\begin{example}[A portfolio of prefabricated-house investments] \label{example:continuous-prefabricated-house} 
Consider now a businessperson who can invest in several lines of cheap prefabricated housing. These houses are most profitable in bad aggregate states, when households substitute toward inexpensive housing, and are less profitable in good aggregate states, when households demand higher-quality housing. Thus, these projects act like insurance against poor aggregate conditions.


Assume that background wealth $W$ has a downward-log-convex density with the support required by \Cref{theorem:insurance}. A convenient example is a reflected Lomax (Pareto Type II) distribution. Let $Y$ have density
$$
g(y)=\frac{\eta c^{\eta}}{(c+y)^{\eta+1}},\qquad y\geq0,
$$
where $c>0$ and $\eta>1$, and set $W=\omega-Y$. Then $W$ has support $(-\infty,\omega]$ and density
$$
f(w)=\frac{\eta c^{\eta}}{(c+\omega-w)^{\eta+1}},\qquad w\leq\omega.
$$
Moreover,
$$
\log f(w)=\log(\eta c^{\eta})-(\eta+1)\log(c+\omega-w),
$$
so
$$
\frac{\mathrm d^2}{\mathrm dw^2}\log f(w)=\frac{\eta+1}{(c+\omega-w)^2}>0
$$
on $(-\infty,\omega)$. Hence $f$ is downward log convex. The restriction $\eta>1$ also gives a finite first moment, as required in the general model. This distribution preserves the intended fat-left-tail interpretation of the insurance example while satisfying the support assumptions of the theorem.

Let $p_{i}:\mathbb{R}\rightarrow \mathbb{R}$ denote the net payoff from the (i)-th prefabricated-house project. Assume that each $p_{i}$ is nonincreasing, $F$-integrable, and has the same crossing point $w^{\ast }<\omega$: $p_{i}(w)\geq 0$ \quad \text{for } $w\leq w^{\ast }$, $p_{i}(w)\leq 0$ \quad  \text{for }$w\geq w^{\ast }$. A simple parametric specification is $p_{i}(w)=a_{i}(w^{\ast }-w),$ $a_{i}>0$. Let $b_{P}^{k}(w)=\sum_{i=1}^{k}p_{i}(w)$, $k=1,\ldots ,n,$ and define $B_{P}=\{\mathbf{0},b_{P}^{1},\ldots ,b_{P}^{n}\}$. Assume in addition that $w+b_P^k(w)$ is nondecreasing for every $k$. Thus $b_{P}^{1}$ corresponds to investing in one prefabricated-house line, $b_{P}^{2}$ adds a second line, and so on.

Each $b_P^k$ is nonincreasing, and $w+b_P^k(w)$ is nondecreasing by assumption. In the linear specification, the latter property follows from $A_n:=\sum_{i=1}^na_i\leq1$, since $A_k\leq A_n$ and $w+b_P^k(w)=(1-A_k)w+A_kw^*$. Finally, the common-crossing assumption implies that, for $k'\geq k$, $b_P^{k'}-b_P^k$ is nonnegative below $w^*$ and nonpositive above it. Hence $B_P$ is an insurance menu under the order $b_P^{k'}\succeq b_P^k$ if and only if $k'\geq k$.

To verify ZAZI for a nontrivial ordered family, specialize to $w^*=0$ and assume $\omega>c/(\eta-1)$, so $\E[W]>0$. For $u_q(m)=m-q(-m)^+$, define $\bar q=\E[W]/\E[(-W)^+]>0$. Each $u_q$ is increasing, concave, and Lipschitz, and $u_{q'}$ is more risk averse than $u_q$ whenever $q'\geq q$. Since $W+b_P^k(W)=(1-A_k)W$ and $u_q$ is positively homogeneous,
$$
\E[u_q(W+b_P^k(W))]=(1-A_k)\E[u_q(W)]\leq\E[u_q(W)]
$$
for every $q\in[0,\bar q]$. Hence ZAZI holds throughout this family. By \Cref{theorem:insurance}, if $0\leq q\leq q'\leq\bar q$, then $M(u_{q'})\geq M(u_q)$.
\end{example}

\subsection{Mixed Asset Menus}\label{sec:mixed-menus}

The preceding applications concern menus ordered entirely as investments or entirely as insurance. A mixed menu raises a distinct question. The following example does not compare willingness to pay for information; conditional on purchasing full information, it shows how risk aversion changes whether the decision maker selects an investment or an insurance opportunity.

\begin{example}\label{example investment and insurance} 
Suppose that a businessperson with CARA utility $u_r(m)=-e^{-rm}$, $r>0$, is deciding between the two assets described in Examples~\ref{example investment} and~\ref{example insurance} in \Cref{sec:motivating-examples}. As in those examples, there are two states of the world: high, and low. The payoff from not opting for any asset is $\left( k_{H},k_{L}\right) $, with $k_{H}>k_{L}$. 

The net payoff from undertaking the electric-car project is $H_E>0$ in state $H$ and $-L_E<0$ in state $L$. The net payoff from undertaking the prefabricated-house project is $-H_P<0$ in state $H$ and $L_P>0$ in state $L$, with $k_L+L_P<k_H-H_P$.
Because she can oversee only one company, she selects one project before observing a fully revealing signal. After observing the state, she may undertake the selected project or decline it at zero payoff.

Assume that the prior probability of the state being low, denoted by $p$, is such that: 
\begin{equation*}
(1-p)H_E<pL_E
\qquad\text{and}\qquad
pL_P<(1-p)H_P.
\end{equation*}

These inequalities mean that neither project is attractive to an uninformed risk-neutral decision maker. However, they do not imply ZAZI for every $r>0$. Precisely, there is a unique $\bar r_P>0$ with the property that the status quo is optimal relative to the full uninformed menu exactly when $0<r\leq\bar r_P$. The project-selection threshold $r_0$ is defined conditional on purchasing information and need not lie within this ZAZI range. In this case, the result does not establish a switch in project choice among all agents satisfying ZAZI.

Conditional on purchasing information, the price $I$ multiplies the exponential losses from $E$ and $P$ by the common factor $e^{rI}$. Hence $E$ is weakly preferred to $P$ if and only if $\rho(r)\geq1$, where $\rho$ is defined below.

\begin{lemma}
\label{lemma:example3}
Let $0<p<1$, $H_E,L_P>0$, and $k_H-k_L>L_P$.
Define, for $r>0$,
$$
\rho(r):=\frac{1-p}{p}e^{-(k_H-k_L)r}\frac{1-e^{-H_Er}}{1-e^{-L_Pr}}.
$$
Then $\rho$ is strictly decreasing. If
\begin{equation}
\label{condition-PE}
\frac{1-p}{p}\frac{H_E}{L_P}>1,
\end{equation}
then $\rho(r)=1$ has a unique solution $r_0>0$, with $\rho(r)>1$ for $0<r<r_0$ and $\rho(r)<1$ for $r>r_0$. If inequality~\eqref{condition-PE} is reversed weakly, then $\rho(r)<1$ for every $r>0$.
\end{lemma}

After full revelation, project $E$ yields $H_E$ in state $H$ and zero in state $L$, while project $P$ yields zero in state $H$ and $L_P$ in state $L$. Thus \eqref{condition-PE} says that a risk-neutral decision maker strictly prefers $E$ to $P$. Conditional on purchasing information, agents with $0<r<r_0$ choose $E$, agents with $r>r_0$ choose $P$, and the agent with $r=r_0$ is indifferent. If the inequality in \eqref{condition-PE} is reversed weakly, every agent with $r>0$ chooses $P$ over $E$.

\end{example}

\section{Related Literature}\label{sec:related-literature}

Our analysis is related to four strands of literature: the classical theory of risk and comparative risk aversion, work on risky investment and insurance decisions, the literature on information acquisition and the value of information, and the macroeconomics and finance literature emphasizing tail risk.

\paragraph{Risk, risk aversion, and risky decisions.}
The preference comparison used in this paper evolves from the classical expected-utility theory of risk aversion. \textcite{pratt1964risk} and \textcite{arrow1971riskbearing} develop the local and global notions of risk aversion that underlie comparisons of certainty equivalents and willingness to pay for insurance. \textcite{rothschildstiglitz1970,rothschildstiglitz1971} provide the canonical distributional notion of an increase in risk, while \textcite{diamondstiglitz1974} study how increases in risk interact with increases in risk aversion. These papers are directly relevant to ours because our proofs compare terminal-wealth distributions and ask how those comparisons are ranked as preferences become more risk averse.

Because initial wealth is random in our model, \textcite{KihlstromRomerWilliams1981} provides a useful benchmark: it identifies conditions under which Arrow--Pratt comparisons of risk premia survive independent random initial wealth. Our setting allows asset payoffs to depend on the background-wealth state and compares information prices rather than risk premia.

Related work develops stronger comparative notions of risk aversion and the effects of background risk. See, among others, \textcite{ross1981,jewitt1987,jewitt1989,gollier1996risk} and the synthesis in \textcite{gollier2001economics}.

Recent work provides behavioral foundations for risk aversion based directly on insurance choices. \textcite{MaccheroniMarinacciWangWu2025} show that weak risk aversion is equivalent to propensity to full insurance, while strong risk aversion is equivalent to propensity to several forms of partial insurance and hedging. They obtain parallel comparative results. For example, Yaari's ordering corresponds to comparative propensity to full insurance, whereas Ross's stronger ordering corresponds to comparative propensity to partial insurance. \textcite{CoteWangWu2025} develop this connection further using the notion of risk--insurance parity. They characterize the classes of indemnity functions associated with weak and strong risk aversion. Then, they introduce intermediate attitudes which are characterized by propensity to deductible-only and limit-only contracts.

These characterizations compare an insurance contract with equidistributed alternatives with different dependence on background risk. We study willingness to pay for an experiment that allows asset choice to depend on the observed signal. As a result, the terminal wealth with and without information may have different marginal distributions. On top of that \textcite{MaccheroniMarinacciWangWu2025}'s characterization of partial-insurance propensity uses Ross's strong ordering. In contrast, our insurance result uses a standard concave-transformation ordering with some additional assumptions on the experiment, the asset menu, and the tails of background risk.

A related literature studies optimal insurance when there is background risk. \textcite{ChiWei2020} characterize optimal indemnities under general dependence between insurable and background risks. They show that the resulting contracts need not be piecewise-linear. \textcite{HinckSteinorth2023} consider loss-dependent background risk. They are able to identify conditions involving risk vulnerability, prudence, and higher-order risk attitudes under which background risk raises insurance demand. We have a setting where background risk also affects the value of insurance. But that risk is itself the object about which information is acquired. The agent purchases a signal about background wealth and subsequently selects an asset. This means that the comparative statics are governed by the interaction between information,  asset choice, and the distributional geometry of background wealth.

Other recent work studies insurance menus and insurance demand when expected-utility does not hold. \textcite{GershkovMoldovanuStrackZhang2023} assume dual utility and adverse selection. They derive optimal menus consisting of deductible--premium or coverage-limit--premium pairs. \textcite{JaspersenPeterRagin2023} in turn analyze the effect of probability weighting on coinsurance, deductible choice, and insurance against low-probability losses. They identify conditions under which probability weighting and utility curvature act as substitutes. More broadly, \textcite{RiegerFels2024} show that state-dependent financial needs, rather than just risk aversion, may lead to insurance demand. In our case, we maintain the expected utility assumption. This allows us to isolate how comparative risk aversion affects information acquisition when information can be used either to invest in profitable risky assets, or to target insurance toward adverse states.

Risk aversion also plays a key role in investment and production decisions under uncertainty. For example, \textcite{sandmo1971} studies the behavior of a risk-averse competitive firm under price uncertainty, while portfolio-choice work such as \textcite{samuelson1969lifetime} analyzes how risky investment responds to preferences and investment opportunities. Our distinction between investment and insurance menus is in the same spirit, but our comparative static concerns the value of acquiring information before the asset choice. In addition, our equilibrium motivation is related to the classic information-market analysis of \textcite{GrossmanStiglitz1980}. When acquiring information is costly, risk preferences can affect which agents become informed and therefore, in richer market models, who bears risk and how information is incorporated into prices. The insurance side is similarly connected to the role of information in insurance markets emphasized by \textcite{RothschildStiglitz1976}.

\paragraph{Risk aversion and the value of information.}

A separate literature asks directly how the value or demand for information changes with risk and risk aversion. \textcite{gould1974risk} studies how changes in payoff risk affect information value, while \textcite{hilton1981determinants} shows that attributes of the decision problem and decision maker need not have a uniform directional effect. \textcite{freixas1984risk} obtains decreasing information demand under specific assumptions, and \textcite{willinger1989risk} identifies a local CARA condition under which the sign depends on how information changes the risk of the optimal action. \textcite{alepuz1995note} separates the ex post benefit of reduced uncertainty from the ex ante riskiness of returns to information, while \textcite{eeckhoudt2000risk} illustrates why the two effects need not deliver a uniform sign. \textcite{nadiminti1996risk} derives conditions for either sign and emphasizes that the method of payment for information matters. \textcite{kihlstrom2022risk} surveys decreasing investment-side results in \textcite{CabralesGossnerSerranoEntropyAER2013,cabrales2017normalized} and the CARA portfolio analysis of Losq and Sobti.

Our earlier work, \textcite{CabralesGossnerSerranoEntropyAER2013,cabrales2017normalized}, constructs complete orderings of information based on investors' willingness to pay and obtains an inverse relation between the value of information and risk aversion in an environment with a safe outside option. The current paper explains why one should expect that decreasing relation for investment menus but it need not extend to insurance-like opportunities. More  generally, we connect a collection of apparently conflicting examples using conditions stated in terms of two primitives with a simple economic interpretation. Our focus is on asset protection or expansion of risk, on the one hand, and the tail geometry of background risk on the other.

\paragraph{Tails in macroeconomics and finance.}

Our assumptions about the tails of background risk connect the paper to a large macro-finance literature. In those papers, low-probability, severe outcomes have first-order effects on quantities and asset prices. \textcite{rietz1988} showed that rare market crash events can substantially alter equilibrium risk premia. \textcite{barro2006rare} used twentieth-century macroeconomic disasters to revive and quantify the rare-disaster mechanism. He showed that other puzzles could be solved using the same mechanism. \textcite{gabaix2012rare} showed how time-varying disaster severity can account for a broad set of macro-finance facts. \textcite{gourio2012disaster} focuses on business cycles. He showed that disaster risk variation can move output, employment, investment, and asset prices. There is a complementary literature studying the endogenous origins of aggregate tail events. \textcite{acemoglu2017tail}, for example, show how sectoral heterogeneity and input-output linkages can generate macroeconomic tail risks. This is true even when ordinary fluctuations are close to Gaussian away from the tails.

We use tails in a different but complementary way. We do not introduce disaster states to explain aggregate moments or asset prices. Rather, we impose restrictions on the distribution of background wealth and then show the impact of those restrictions on the comparative statics of information demand. 
Log concavity captures thin-tails so that a fixed downward movement of wealth has a strong  effect on the lower-tail. Downward log convexity, on the other hand, captures a fat-left-tail so that the targeting benefit of insurance can dominate that translation. This means that for us tail behavior is a primitive condition affecting the interaction of information, asset payoffs, and risk aversion.

\paragraph{Contribution relative to the literature.}
The combination of these strands clarifies our paper's contribution. The classical literature on risk clarifies the ordering of preferences and distributions. The literature on investment and insurance is the basis for the economic interpretation of risk-increasing and risk-protecting assets. The literature on information anticipated that the comparative static of information value with respect to risk aversion could have either sign. Finally, the macro-finance literature hints at the reasons why the tails of economic risk can be decisive. Our results connect these ideas. We provide a single general framework where asset structure and tail geometry codetermine the direction of the information-acquisition comparative static.

\section{Conclusion}\label{sec:conclusion}

We study how comparative risk aversion changes willingness to pay for information acquired before an asset choice under background risk. For finite MLR-ordered experiments and decision makers satisfying ZAZI, willingness to pay decreases for investment menus under log-concave background risk and increases for insurance menus under downward-log-convex background risk. The common mechanism is the interaction between the uniform wealth loss from the information price and the signal-contingent reallocation generated by asset choice; our state-dependent single-crossing result provides the aggregation step linking this mechanism to the two comparative statics. The investment and insurance reversals show that the tail restrictions have substantive force without providing a necessity characterization. These results isolate a primitive determinant of information acquisition that can be incorporated into market models, whose equilibrium analysis remains outside the scope of the paper.

Our framework assumes that asset payoffs are deterministic functions of background wealth. A natural extension would allow their conditional distributions to vary with background wealth, thereby accommodating more general dependence between asset returns and background risk. This extension also lies beyond the scope of the paper, and we hope to contribute to this question in future research.


\printbibliography

\newpage
\phantomsection
\label{app:start}
\begin{appendices}
\crefname{appendix}{Appendix}{Appendices}
\Crefname{appendix}{Appendix}{Appendices}
\crefalias{section}{appendix}
\crefalias{subsection}{appendix}

\section{Comparative Statics Under Full Revelation}\label{sec:full revelation}
We now spell out the logic of the comparative statics when the signal is fully revealing. 

Let $W$ be an integrable random variable with distribution $F$, and let $B$ be a finite set of integrable asset payoffs $b:\mathbb R\to\mathbb R$ containing the status quo random variable $\mathbf{0}$ which equals 0 a.s. The decision maker has a strictly increasing utility function $u\in U$. We restrict attention to utilities for which the status quo is optimal without information:
$
\mathbf{0}\in\argmax_{b\in B}\mathbb E[u(W+b(W))].
$
If she pays $\mu\ge0$ to observe $W$ before choosing an asset, monotonicity of $u$ implies that she chooses an asset with maximal payoff at the realized state. Let
$
c(w)=\max_{b\in B}b(w).
$
Her terminal wealth with complete information is therefore
$
T_\mu(W)=W+c(W)-\mu.
$
We define her willingness to pay for complete information by:
$$
M^{FI}(u)
=
\sup\{\mu\ge0:\mathbb E[u(T_\mu(W))]\ge\mathbb E[u(W)]\}.
$$

An asset is an \emph{investment} if $b$ is nondecreasing. An asset is an
\emph{insurance} if $b$ is nonincreasing and $w+b(w)$ is nondecreasing.

\begin{proposition}
\label{prop:complete-information}
Under complete information:
\begin{enumerate}[label=(\roman*)]
\item if every asset in $B$ is an investment, willingness to pay for information decreases with risk aversion;
\item if every asset in $B$ is an insurance, willingness to pay for information increases with risk aversion.
\end{enumerate}
These conclusions hold for any distribution of background wealth.
\end{proposition}

\begin{proof}
Suppose first that every asset is an investment. Then $c$ is nondecreasing, so $T_\mu$ is increasing and
$
T_\mu(w)-w=c(w)-\mu
$
is nondecreasing. Consequently, the distributions of $T_\mu(W)$ and $W$ have the following \textit{single-crossing property}: once the CDF of $W$ lies above that of $T_\mu(W)$, it remains above it at all higher wealth levels. By the standard single-crossing characterization of comparative risk aversion of \textcite{jewitt1987}, if a more risk-averse decision maker prefers $T_\mu(W)$ to $W$, then every less risk-averse decision maker does as well. Hence $M^{FI}(u)$ decreases with risk aversion.

Suppose instead that every asset is an insurance. Then $c$ is nonincreasing, while
$
w+c(w)=\max_{b\in B}\{w+b(w)\}
$
is nondecreasing. Thus $T_\mu$ is increasing and $T_\mu(w)-w$ is nonincreasing. The preceding single-crossing comparison is reversed: once the CDF of $T_\mu(W)$ lies above that of $W$, it remains above it at all higher wealth levels. Hence, if a less risk-averse decision maker prefers $T_\mu(W)$ to $W$, every more risk-averse decision maker does as well. Therefore $M^{FI}(u)$ increases with risk aversion.
\end{proof}

Complete information eliminates the difficulty created by background risk in the general model.  For investments, the gain from adapting the asset choice is increasing in background wealth; for insurance, it is decreasing; this is enough to ensure single crossing. With imperfect information, asset choice instead depends on a noisy signal, and the resulting conditional comparisons must be aggregated across background-wealth states. In that case, further distributional assumptions are needed to ensure aggregation preserves comparative statics.


\section{Proof of Proposition~\ref{prop:mustar}}
\label{app:prop:mustar}

We begin by establishing existence of optimal strategies. 
\begin{lemma}
    \label{lemma:opt}
    Let $u \in U$. The integral $\int u(w + b(w)) \,\mathrm{d}F(w)$ exists and lies in $\mathbb{R} \cup \{-\infty\}$ for all $b \in B$, and the maximization problem 
    \begin{equation*}
        \max_{b \in B} \int u(w + b(w)) \,\mathrm{d}F(w)
    \end{equation*}
    admits a solution. Given $x \in X$, the integral $\int u(w + b(w)-\mu) f_x(w) \,\mathrm{d}w$ exists and lies in $\mathbb{R} \cup \{-\infty\}$ for all $b \in B$, and the maximization problem 
    \begin{equation*}
        \max_{b \in B} \int u(w + b(w)-\mu) f_x(w) \,\mathrm{d}w
    \end{equation*}
    admits a solution.
\end{lemma}

\begin{proof}[Proof of \Cref{lemma:opt}]
    For each $x \in X$, let $F_x$ be the CDF of $f_x$. Fix $F' \in \{F\} \cup\{F_x : x \in X\}$ and $\mu \ge 0$.
    It suffices to show that, given any $b \in B$, the integral $\int u (w + b(w)-\mu) \,\mathrm{d}F'(w)$ exists and lies in $\mathbb{R} \cup \{-\infty\}$, and that the map $b \mapsto \int u(w+b(w)-\mu)\,\mathrm{d}F'(w)$ is upper-semicontinuous with respect to the $L^1(F)$ norm, as this guarantees that the maxmisation problem \begin{equation*}
        \max_{b \in B} \int u(w+b(w)-\mu)\,\mathrm{d}F'(w)
    \end{equation*} 
    admits a solution, since $B$ is compact.

    Any $F$-integrable map is also $F'$-integrable and convergence in $L^1(F)$ implies convergence in $L^1(F')$ since, given any $x \in X$,  $f_x(w)\,\mathrm{d}w = h_w(x) \,\mathrm{d}F(w)/h(x)$ and $w \mapsto h_w(x)$ is bounded.

    Choose $\delta > \limsup_{m \to \infty} u(m)/m$ and $m' > 0$ such that $u(m)/m \le \delta$ for all $m \ge m'$.
    Since $u$ is increasing, $u$ lies pointwise weakly below the map $\bar u(m) = \max\{u(m'),\delta m\}$. Given $b \in B$, the map $w \mapsto \bar u(w + b(w))$ is $F$-integrable since $b$ and the identity function are; hence it is $F'$-integrable. Then, the integral $\int u(w + b(w)-\mu) \,\mathrm{d}F'(w)$ exists and lies in $\mathbb{R} \cup \{-\infty\}$, since $u(m-\mu) \le \bar u(m)$ for all $m \in \mathbb{R}$.

    To establish upper-semicontinuity, suppose that $b_n \to b$ in $L^1(F)$ and that $\int u(w+b_n(w)-\mu)\,\mathrm{d}F'(w)$ converges.
    Since $\bar u$ is Lipschitz, $\bar u(w+b_n(w))\to\bar u(w+b(w))$ in $L^1(F')$, and their integrals converge.
    Moreover, after passing to a subsequence if necessary, we may assume without loss that $b_n \to b$ $F'$-a.e. Let $v : \mathbb{R} \to \mathbb{R}$ be given by $v(m) = \bar u(m) - u(m-\mu)$. $v$ is continuous, so that $v(w+b_n(w)) \to v(w + b(w))$ $F'$-a.e. 
    Then 
    \begin{align*}
        &\lim_{n \to \infty} \int u(w+b_n(w)-\mu) \,\mathrm{d}F'(w) 
        \\&\quad \le \lim_{n \to \infty} \int \bar u(w + b_n(w)) \,\mathrm{d}F'(w) - \liminf_{n\to\infty} \int v(w + b_n(w)) \,\mathrm{d}F'(w) 
        \\&\quad\le \int \bar u(w + b(w)) \,\mathrm{d}F'(w) - \int \lim_{n \to \infty} v(w + b_n(w)) \,\mathrm{d}F'(w) 
        \\&\quad = \int u(w+b(w)-\mu)\,\mathrm{d}F'(w)
    \end{align*}
    where the inequality follows from Fatou's lemma, since $v$ is positive.
\end{proof}

\begin{proof}[Proof of \Cref{prop:mustar}]
For all $u \in U$ and $\mu \ge 0$,  
\begin{equation*}
    V(u,\mu) = \max_{\beta \in B^X }\int u \,\mathrm{d}G^\mu_\beta
\end{equation*}
is well defined by \Cref{lemma:opt}.
Fix $u \in U$.
Since $u$ is increasing and $F$-integrable, $\int u \,\mathrm{d}F > \lim_{m \to -\infty} u(m)$.
Then, it suffices to show that $V(u,\mu) < V(u,\mu')$ for all $0 \le \mu' < \mu$ such that $V(u,\mu)> -\infty$, that $V(u,0) \ge \int u \,\mathrm{d}F$ and that $\lim_{\mu \to \infty} V(u,\mu) \le \lim_{m \to -\infty} u(m)$.
Indeed, this guarantees that $M(u) = \sup\{ \mu \in \mathbb{R}: V(u,\mu) \ge \int u \,\mathrm{d}F\}$ is well defined and lies in $\mathbb{R}$, that $V(u,\mu) > \int u \,\mathrm{d}F$ for all $0\le \mu < M(u)$, and $V(u,\mu) < \int u \,\mathrm{d}F$ for all $\mu > M(u)$.

Fix $0 \le \mu' < \mu$ such that $V(u,\mu) > -\infty$, and let $\beta$ achieve the maximum at $\mu$.
Since $X$ is finite, $V(u,\mu) \in \mathbb{R}$ by \Cref{lemma:opt}.
Then, 
\begin{equation*}
    V(u,\mu) = \int u \,\mathrm{d}G^\mu_\beta < \int u \,\mathrm{d}G^{\mu'}_\beta \le V(u,\mu'),
\end{equation*}
where the strict inequality holds since $u$ is increasing.

Since $\mathbf{0} \in B$, $V(u,0) \ge \int u \,\mathrm{d}F$. To prove that $\lim_{\mu\to\infty}V(u,\mu)=\lim_{m\to-\infty}u(m)$, fix a sequence $\mu_n\to\infty$ and choose $\beta^n$ optimal at $\mu_n$. Set $E_x^n=\{w:w+\beta_x^n(w)>\mu_n/2\}$ for all $x \in X$ and $n \in \mathbb{N}$, and $E^n=\bigcup_{x\in X}E_x^n$. Let $c = \sup\left\{\int |w+\beta_x^n(w)|\,\mathrm dF(w):n\in\mathbb N,\ x\in X\right\}$ and $\nu_F$ be the probability measure induced with $F$.
Since $\int |w| \,\mathrm{d}F < \infty$ and $B$ is compact in $L^1(F)$, the family of maps $(w \mapsto w + b(w) : b \in B)$ is uniformly integrable; that is, 
$c < \infty$ and for any sequence of measurable sets $(D^n)_{n \in \mathbb{N}}$ such that $\nu_F(D^n) \to 0$, $\sup_{b \in B} \int_{D^n} |w + b(w)| \,\mathrm{d}F(w) \to 0$.
Since $c < \infty$, Markov's inequality yields $\nu_F(E^n)\le2|X|c/\mu_n\to0$.
Choose $\delta>0$ and $C\ge0$ such that $u(m)\le C+\delta \max\{0,m\}$ for every $m\in\mathbb R$. Then

\begin{equation*}
    V(u,\mu_n)\le u(-\mu_n/2)(1-\nu_F(E^n))+C\nu_F(E^n)+\delta\sum_{x\in X}\int_{E^n}\max\{0,w+\beta_x^n(w)\}\,\mathrm dF(w).
\end{equation*}
The second and third terms vanish as $n \to \infty$ (the latter by uniform integrability), while the first tends to $\lim_{m\to-\infty}u(m)$. Since $\mu_n$ was arbitrary, the desired limit follows.
\end{proof}

\section{Proofs for \texorpdfstring{\Cref{sec:proof-roadmap,sec:single-crossing-aggregation}}{Sections 4.3 and 4.4}}\label{app:proof-roadmap}

\begin{proof}[Proof of \Cref{lemma:ptp_char}]
By Theorem 1 of \textcite{jewitt1989}, for any $u_1,u_2 \in U$ and any strategy $\beta \in B^X$, 
\begin{equation}
    \label{eq:ptp_decr}
        \int u_1 \,\mathrm{d}G^\mu_\beta > \int u_1 \,\mathrm{d}F \quad \text{implies} \quad \int u_2 \,\mathrm{d}G^\mu_\beta \ge \int u_2 \,\mathrm{d}F
\end{equation}
if either of the following hold: (i) $u_1$ is more risk-averse than $u_2$ and $F-G^\mu_\beta$ is single crossing, or (ii) $u_2$ is more risk-averse than $u_1$ and $G^\mu_\beta-F$ is single crossing.
 
Now, to prove the first part, let $u_1,u_2\in U$ satisfy \eqref{eq:ZAZI}, with $u_1$ more risk averse than $u_2$. If $M(u_1)=0$, then $M(u_1)\le M(u_2)$. Otherwise, fix $0\le\mu<M(u_1)$ and let $\beta$ be optimal for $u_1$ at price $\mu$ with $F-G^\mu_\beta$ single crossing. Since $\mu<M(u_1)$, $u_1$ strictly prefers purchasing information by \Cref{prop:mustar}, so $\int u_1\,\mathrm dG^\mu_\beta>\int u_1\,\mathrm dF$. By \eqref{eq:ptp_decr}, $\int u_2\,\mathrm dG^\mu_\beta\ge\int u_2\,\mathrm dF$, and therefore $u_2$ weakly prefers purchasing information at price $\mu$. Thus $\mu \le M(u_2)$ by \Cref{prop:mustar}. Since this holds for every $\mu<M(u_1)$, $M(u_1)\le M(u_2)$.

The second part follows by the same argument, taking $u_2$ to be more (instead of less) risk averse than $u_1$, and choosing $\beta$ such that $G^\mu_\beta-F$ is  single crossing (instead of $F-G^\mu_\beta$).
\end{proof}

\begin{proof}[Proof of \Cref{lemma:mon}]
Fix $u\in U$ and $\mu \ge 0$, and set
$$
X=\{x_1<\cdots<x_K\}
$$
for some $K \in \mathbb{N}$.
By \Cref{lemma:opt}, $u$ admits an optimal strategy $\beta$ and $\int u \,\mathrm{d}G^\mu_\beta \in \mathbb{R} \cup \{-\infty\}$.
The result is trivial if either $K = 1$ or $\int u \,\mathrm{d}G^\mu_\beta = -\infty$, hence suppose that $K \ge 2$ and $\int u \,\mathrm{d}G^\mu_\beta$ is finite.

Choose $y_0<\cdots<y_K$ in $\mathbb{R}$ and set $Y = (y_0,y_K]$. For each $w\in\mathbb R$, let $P(\cdot,w)$ be the CDF with support $Y$ and density
$$
p_w(y)
=
\frac{h_w(x_k)}{y_k-y_{k-1}}
\qquad
\text{for }y\in(y_{k-1},y_k].
$$
Since the experiment $\alpha$ is TP2, the map $(w,y)\mapsto p_w(y)$ is TP2.

Let
$$
\{b_1,\ldots,b_n\}
=
\{\beta_x:x\in X\},
\qquad
b_1\preceq\cdots\preceq b_n,
$$
and define $L_i : \mathbb{R} \to \mathbb{R}$ by 
$$
L_i(w)=-u(w+b_i(w)-\mu)
$$
for each $i \in \{1,\dots,n\}$.
Since $u$ is increasing and $b_{i+1}-b_i$ is single crossing,
$L_i-L_{i+1}$ is single crossing for each $i<n$.

View $\beta$ as the decision procedure $Y \to \{1,\dots,n\}$ that chooses, for each $k \in \{1,\dots,K\}$ and $y \in (y_{k-1},y_k]$, $i \in \{1,\dots,n\}$ such that $\beta_{x_k} = b_i$.
Since $P(\cdot|w)$ admits a density for all $w \in \mathbb{R}$, Theorem~1 of
\textcite{karlin1956} delivers a (deterministic) nondecreasing procedure $\phi : Y \to \{1,\dots,n\}$ whose risk
is weakly lower at every $w \in \mathbb{R}$, i.e., 
\begin{equation*}
    \sum_{k = 1}^K \frac{h_w(x_k)}{y_k-y_{k-1}} \int_{y_{k-1}}^{y_k} u(w+b_{\phi(y)}(w)-\mu) \,\mathrm{d}y \ge \sum_{k = 1}^K h_w(x_k)u(w+\beta_{x_k}(w)-\mu).
\end{equation*}
Integrating with respect
to $F$ gives
\begin{align}
\nonumber
&\sum_{k=1}^K
\frac{h(x_k)}{y_k-y_{k-1}}
\int_{y_{k-1}}^{y_k}
\int u(w+b_{\phi(y)}(w)-\mu)\,\mathrm dF_{x_k}(w)\,\mathrm dy
\\
\label{eq:mon_aggr}
&\qquad\ge
\sum_{k=1}^K h(x_k)
\int u(w+\beta_{x_k}(w)-\mu)\,\mathrm dF_{x_k}(w).
\end{align}
For every $k\in \{1,\dots,K\}$ and $y\in(y_{k-1},y_k]$,
\begin{equation}
\label{eq:mon_y}
\int u(w+b_{\phi(y)}(w)-\mu)\,\mathrm dF_{x_k}(w)
\le
\int u(w+\beta_{x_k}(w)-\mu)\,\mathrm dF_{x_k}(w),
\end{equation}
because $\beta$ is optimal for $u$, $h(x_k) > 0$ and $\int u \,\mathrm{d}G^\mu_\beta$ is finite. Then \eqref{eq:mon_aggr} implies that, for each $k \in \{1,\dots,K\}$, \eqref{eq:mon_y}
holds with equality for a.e.\ $y \in (y_{k-1},y_k]$, since $h(x_k) > 0$. 
Consequently, we may choose $z_k \in (y_{k-1},y_k]$ for each $k \in \{1,\dots,K\}$ such that the strategy $\beta' \in B^X$ given by $\beta'_{x_k} = b_{\phi(z_k)}$ satisfies $\int u \,\mathrm{d}G^\mu_{\beta'} = \int u \,\mathrm{d}G^\mu_\beta$.
Thus, $\beta'$ is optimal for $u$, and it is monotone since $\phi$ is nondecreasing. 
\end{proof}
The proof of \Cref{lemma:sc_random} relies on \Cref{lemma:sc_incr} from \Cref{sec:single-crossing-aggregation}. We begin by stating the classical result from \textcite{karlin1956}, then use it to prove \Cref{lemma:sc_incr}, and finally we prove \Cref{lemma:sc_random}.

\begin{lemma}[Karlin and Rubin, 1956]
\label{lemma:sc}
Let $\phi:\mathbb R\to\mathbb R$ be single crossing and let $\kappa:\mathbb R^2\to\mathbb R_+$ be TP2. Then
$$
m\longmapsto
\int\phi(\ell)\kappa(\ell,m)\,\mathrm{d}\ell
$$
is single crossing, provided $\ell\mapsto\phi(\ell)\kappa(\ell,m)$ is integrable for every $m\in\mathbb R$.
\end{lemma}

\Cref{lemma:sc_incr} yields \Cref{lemma:sc_random}\ref{item:sc_random:logconcave} by taking the aggregation kernel from a log-concave density of background wealth. This provides the aggregation step in the proof of \Cref{theorem:investment}. To obtain \Cref{lemma:sc_random}\ref{item:sc_random:logconvex}, which involves downward-log-convex background risk, \Cref{lemma:sc_incr} is applied only to a bounded-above interval of wealth levels, after extending the density to obtain a TP2 kernel; the remaining region is handled by a direct distributional comparison. Thus \Cref{theorem:investment,theorem:insurance} share the same single-crossing principle but require different aggregation arguments.

\begin{proof}[Proof of \Cref{lemma:sc_incr}]
Define
$$
\Phi(\ell,m):=\int_{\mathbb R}\phi(k,\ell)\kappa(k,m)\,\mathrm dk,\qquad \ell,m\in I.
$$
Fix $\ell,m\in I$ with $\ell\leq m$ and $\Phi(\ell,\ell)>0$. By \Cref{lemma:sc}, the map $m'\mapsto\Phi(\ell,m')$ is single crossing, so $\Phi(\ell,m)\geq0$. Moreover, $\ell>\underline m$, since otherwise $\phi(k,\ell)\leq0$ for every $k$, contradicting $\Phi(\ell,\ell)>0$. If $m<\bar m$, then $\underline m<\ell\leq m<\bar m$, and monotonicity in the second argument of $\phi$ gives
$$
\Phi(m,m)\geq\Phi(\ell,m)\geq0.
$$
If $m\geq\bar m$, then $\phi(k,m)\geq0$ for every $k$, so again $\Phi(m,m)\geq0$. Thus $\Phi(\ell,\ell)>0$ and $\ell\leq m$ imply $\Phi(m,m)\geq0$, proving that $m\mapsto\Phi(m,m)$ is single crossing.
\end{proof}
\begin{proof}[Proof of \Cref{lemma:sc_random}]
We establish each of the two parts by relying on \Cref{lemma:sc_incr}.
For \ref{item:sc_random:logconcave}, let $f_W$ denote
the density of $W$, and define
$$
\phi(\ell,m)
=
F_{Y^1|W}(\ell|m-\ell)-F_{Y^2|W}(\ell|m-\ell)
$$

and
$$
\kappa(\ell,m)=f_W(m-\ell).
$$
The integrability condition in \Cref{lemma:sc_incr} holds 
because $|\phi|\le1$.
Moreover, for $i\in\{1,2\}$,
\begin{align*}
F_{Y^i+W}(m)
&=
\int F_{Y^i|W}(m-w|w)f_W(w)\,\mathrm dw\\
&=
\int F_{Y^i|W}(\ell|m-\ell)f_W(m-\ell)\,d\ell,
\end{align*}
where the second equality follows from the change of variable
$\ell=m-w$. Hence
$$
F_{Y^1+W}(m)-F_{Y^2+W}(m)
=
\int\phi(\ell,m)\kappa(\ell,m)\,d\ell.
$$

For every $m$, the map $\phi(\cdot,m)$ is single crossing by hypothesis. For every $\ell$, the map $\phi(\ell,\cdot)$ is
nonpositive on $\{w \in \mathbb{R}: w \le \underline w\}$, nondecreasing on
$(\underline w,\bar w)$, and nonnegative on $\{w \in \mathbb{R} : w \ge \bar w\}$ by the common
hypothesis of the lemma. Since $f_W$ is log concave, $\kappa$ is
TP2. Therefore, \Cref{lemma:sc_incr} implies that
$$
F_{Y^1+W}-F_{Y^2+W}
$$
is single crossing.

It remains to prove \ref{item:sc_random:logconvex}.
Clearly, it suffices to show that $F_{Y^1+W}-F_{Y^2+W}$ is nonnegative on $\{w \in \mathbb{R} : w \ge \bar w\}$ and single crossing on $I = (-\infty,\bar w)$.
To this end, let once again $f_W$ denote the density of $W$, and define $\phi' : \mathbb{R}^2 \to \mathbb{R}$ by 
$$
\phi'(\ell,m)
=
F_{Y^1|W}(-\ell|\ell+m)-F_{Y^2|W}(-\ell|\ell+m).
$$
For $i\in\{1,2\}$,
\begin{align*}
F_{Y^i+W}(m)
&=
\int F_{Y^i|W}(m-w|w)f_W(w)\,\mathrm dw\\
&=
\int F_{Y^i|W}(-\ell|\ell+m)f_W(\ell+m)\,d\ell,
\end{align*}
where the second equality follows from the change of variable
$\ell=w-m$, so that 
$$
F_{Y^1+W}(m)-F_{Y^2+W}(m)
=
\int\phi'(\ell,m)f_W(\ell+m)\,d\ell.
$$
For every $m\ge\bar w$, the common hypothesis of the lemma implies that
$\phi'(\ell,m)\ge0$ for every $\ell$. Hence
$F_{Y^1+W}(m)-F_{Y^2+W}(m)\ge0$ for every $m\ge\bar w$.

It remains to show that $D:=F_{Y^1+W}-F_{Y^2+W}$ is single crossing on
$I:=(-\infty,\bar w)$. Let $q=\log f_W$. Choose $r_n\uparrow\omega$ and a
finite subgradient $a_n\in\partial q(r_n)$, and define the convex function
$$
q_n(t)=
\begin{cases}
q(t),&t\le r_n,\\
q(r_n)+a_n(t-r_n),&t>r_n.
\end{cases}
$$
Set $c_n:=\int_{-\infty}^{\omega}e^{q_n(t)}\,\mathrm dt$ and
$f_W^n(t):=c_n^{-1}e^{q_n(t)}\mathbf 1_{\{t\le\omega\}}$. Since the affine
continuation is a supporting line to $q$, one has $q_n\le q$ on
$(-\infty,\omega)$. Moreover, $q_n(t)=q(t)$ eventually for every fixed
$t<\omega$. Dominated convergence therefore gives $c_n\to1$ and
$\|f_W^n-f_W\|_1\to0$.

Extend $\log f_W^n$ beyond $\omega$ by setting
$\widetilde q_n(t):=q_n(t)-\log c_n$ for every $t\in\mathbb R$, and let
$\kappa_n(\ell,m):=e^{\widetilde q_n(m+\ell)}$. Convexity of
$\widetilde q_n$ implies that $\kappa_n$ is TP2. The support
condition gives $\phi'(\ell,m)=0$ whenever $m\in I$ and
$m+\ell>\omega$. Consequently, for every $m\in I$,
$$
\Phi_n(m)
:=
\int\phi'(\ell,m)\kappa_n(\ell,m)\,d\ell
=
\int\phi'(\ell,m)f_W^n(m+\ell)\,d\ell.
$$
The assumed single-crossing property of $y\mapsto-\phi'(-y,m)$ implies that
$\ell\mapsto\phi'(\ell,m)$ is single crossing, while the common hypotheses
give the required monotonicity in $m$. The integrability condition follows
from $|\phi'|\le1$ and the support property. Hence
\Cref{lemma:sc_incr} implies that every $\Phi_n$ is single crossing on $I$.

Finally, $|\phi'|\le1$ gives
$\sup_{m\in I}|\Phi_n(m)-D(m)|\le\|f_W^n-f_W\|_1\to0$. If $m\le m'$ and
$D(m)>0$, then $\Phi_n(m)>0$ for all sufficiently large $n$, so
$\Phi_n(m')\ge0$ and therefore $D(m')\ge0$. Thus $D$ is single crossing on
$I$. Since $D$ is nonnegative on $[\bar w,\infty)$, it is single crossing
on $\mathbb R$.
\end{proof}
\section{Worked Investment and Insurance Illustrations}\label{sec:worked examples}

The following pair of examples makes the preceding distributional argument explicit. In both cases the signal shows whether background wealth lies above or below the common payoff-crossing point of the asset menu. Because of this the optimal signal-contingent strategy is transparent. In turn, that allows the induced wealth distribution $G^\mu_\beta$ to be computed in closed form.
We recall the one-parameter family
\begin{equation}
    \label{eq:uq-worked-example}
    u_q(m)=m-q(-m)^+,\qquad q\geq0.
\end{equation}
Each $u_q$ is increasing, concave, and Lipschitz. Moreover, if $q'>q$, then $u_{q'}$ is more risk averse than $u_q$ in the sense of \Cref{sec:environment}.
The parameter $q$ therefore orders this family by risk aversion; below we evaluate each crossing at the corresponding willingness-to-pay price.

\medskip
\noindent\emph{Investment.}
Let $W\sim N(0,1)$ and set $w_I^*=1/2$. Consider the investment menu
$$
B_I=\{\lambda b_I:\lambda\in[0,1]\},
\qquad
b_I(w)=w-w_I^*.
$$
The binary signal reveals whether $W<w_I^*$ or $W\geq w_I^*$. This is an MLR-ordered experiment. After the high signal, $b_I(W)\geq0$ throughout the posterior support, whereas after the low signal, $b_I(W)\leq0$. Hence, for every increasing utility and every information price $\mu$, an optimal strategy is
$$
\lambda_L^*=0,
\qquad
\lambda_H^*=1.
$$
Thus higher signals induce weakly larger investment positions, as in \Cref{lemma:mon}. The status quo is also optimal before information for every $q\geq0$: the prior objective is concave in $\lambda$, and its right derivative at $\lambda=0$ is
$$
-\frac12-q\left(\frac{1}{\sqrt{2\pi}}+\frac14\right)<0.
$$
The example therefore satisfies ZAZI for the entire family \eqref{eq:uq-worked-example}.

After paying $\mu$, terminal wealth is
$$
T_I(W)=
\begin{cases}
W-\mu, & W<w_I^*,\\
2W-w_I^*-\mu, & W\geq w_I^*.
\end{cases}
$$
Since this map is increasing and continuous, its distribution can be obtained by inversion. Writing $F=\Phi$ for the standard-normal CDF,
\begin{equation}
    \label{eq:G-investment-worked}
G_I^\mu(m)=
\begin{cases}
F(m+\mu), & m<w_I^*-\mu,\\[1mm]
F\!\left(\dfrac{m+\mu+w_I^*}{2}\right), & m\geq w_I^*-\mu.
\end{cases}
\end{equation}

For $m<w_I^*-\mu$, $F(m)-G_I^\mu(m)\leq0$, with strict inequality if $\mu>0$. For $m\geq w_I^*-\mu$, its sign is the sign of
$$
m-\frac{m+\mu+w_I^*}{2}.
$$
Thus $F-G_I^\mu$ is single crossing. If $\mu>0$, its unique strict sign change occurs at
\begin{equation}
\label{eq:investment-crossing-worked}
m_I^{SC}=w_I^*+\mu,
\end{equation}
from negative to positive. If $\mu=0$, the difference is zero below $w_I^*$ and nonnegative above. This is the orientation required in the first part of \Cref{lemma:ptp_char}.

\medskip
\noindent\emph{Insurance.}
For a mirror-image example satisfying the support assumption of \Cref{theorem:insurance}, let background wealth have the reflected-Lomax distribution
$$
F(w)=\frac{1}{(2-w)^3},\qquad w\leq1,
$$
with $F(w)=1$ for $w>1$. Its density is $f(w)=3(2-w)^{-4}$ on $(-\infty,1]$, so $\log f$ is strictly convex and the support is unbounded below. Set $w_P^*=-1/2$, $a=3/5$, and consider
$$
B_P=\{\lambda b_P:\lambda\in[0,1]\},
\qquad
b_P(w)=a(w_P^*-w).
$$
Because $a\in(0,1)$, $b_P$ is nonincreasing and $w+b_P(w)$ is increasing, so this is an insurance menu. Let the signal reveal whether $W\leq w_P^*$ or $W>w_P^*$. After the low signal the insurance payoff is nonnegative throughout the posterior support, while after the high signal it is nonpositive. Thus an optimal strategy is
$$
\lambda_L^*=1,
\qquad
\lambda_H^*=0,
$$
so higher signals induce weakly smaller insurance positions. For the family \eqref{eq:uq-worked-example}, the prior objective is concave in $\lambda$, and its right derivative at zero is
$$
-\frac35+\frac{3q}{80}.
$$
Hence ZAZI holds for every $q\leq16$, in particular for all the values reported below.

Terminal wealth after purchasing information at price $\mu$ is
$$
T_P(W)=
\begin{cases}
(1-a)W+aw_P^*-\mu, & W\leq w_P^*,\\
W-\mu, & W>w_P^*.
\end{cases}
$$
Again the map is increasing and continuous, and therefore
\begin{equation}
    \label{eq:G-insurance-worked}
G_P^\mu(m)=
\begin{cases}
F\!\left(\dfrac{m+\mu-aw_P^*}{1-a}\right), & m<w_P^*-\mu,\\[2mm]
F(m+\mu), & m\geq w_P^*-\mu.
\end{cases}
\end{equation}
In the lower region, the sign of $G_P^\mu(m)-F(m)$ is the sign of
$$
\frac{m+\mu-aw_P^*}{1-a}-m,
$$
which vanishes at
\begin{equation}
\label{eq:insurance-crossing-worked}
m_P^{SC}=w_P^*-\frac{\mu}{a}.
\end{equation}
Above $w_P^*-\mu$, $G_P^\mu(m)-F(m)=F(m+\mu)-F(m)\geq0$. Thus $G_P^\mu-F$ is single crossing. If $\mu>0$, its unique strict sign change is from negative to positive at $m_P^{SC}$; if $\mu=0$, it is negative below $w_P^*$ and zero above. This is the orientation required in the second part of \Cref{lemma:ptp_char}.

Let $M_I(u_q)$ and $M_P(u_q)$ denote willingness to pay in the investment and insurance examples. The following table reports these values and the associated crossing points. The scalar indifference equations are easily solved numerically.
\begin{center}
\begin{tabular}{c|cc|cc}
$q$ & $M_I(u_q)$ & $m_I^{SC}$ & $M_P(u_q)$ & $m_P^{SC}$\\
\hline
$0.5$ & $0.1563$ & $0.6563$ & $0.0676$ & $-0.6126$\\
$1$   & $0.1296$ & $0.6296$ & $0.0847$ & $-0.6412$\\
$2$   & $0.0970$ & $0.5970$ & $0.1131$ & $-0.6885$\\
$4$   & $0.0648$ & $0.5648$ & $0.1534$ & $-0.7557$
\end{tabular}
\end{center}
Thus willingness to pay falls with risk aversion in the investment example and rises with risk aversion in the insurance example. The optimal strategy is independent of $q$; across rows, $q$ affects the induced distribution only through the price $\mu=M(u_q)$.

These examples use stronger assumptions than the main theorems require. The signal reveals which side of the common payoff-crossing point has occurred, so one action is optimal throughout each posterior support; the distribution $G^\mu_\beta$ can be calculated directly, and the single-crossing comparison follows. As a result, the examples illustrate cleanly the first two elements in the proof: they clarify the role of the monotone strategy and the single-crossing wealth comparison. The tail-shape assumptions play no role in establishing these comparisons.

With more general MLR signals, posterior supports may overlap. The conditional wealth comparisons must then be aggregated across background-wealth realizations to obtain the unconditional comparison. The tail assumptions ensure that this aggregation preserves the relevant single crossing.

\section{Additional Lemmas for Section~\ref{sec:main-proofs}}
\label{app:main-proofs}

The following auxiliary lemmas establish some technical results used in the proofs of the main theorems. The first derives first-order stochastic dominance from TP2. The latter two analyze the conditional wealth distributions induced by monotone strategies.

\begin{lemma}[TP2 and stochastic monotonicity]
\label{lemma:tp2-fosd}
Let $X=\{x_1<\cdots<x_K\}$, and suppose that the kernel
$(w,x)\mapsto h_w(x)$ is TP2. Then, whenever $w\le w'$, the distribution
$h_{w'}$ first-order stochastically dominates $h_w$.
\end{lemma}

\begin{proof}
Fix $w\le w'$ and set $a_i=h_w(x_i)$ and $b_i=h_{w'}(x_i)$. TP2 gives $a_i b_j\ge a_j b_i$ for $i<j$. Hence $b_i>a_i$ implies $b_j\ge a_j$ for every $j>i$, so $(b_i-a_i)_{i=1}^K$ changes sign at most once, from nonpositive to nonnegative. Since $\sum_{i=1}^K(b_i-a_i)=0$, it follows that $\sum_{i=1}^k b_i\le\sum_{i=1}^k a_i$ for every $k<K$, which is the claimed first-order stochastic dominance.
\end{proof}

\begin{lemma}\label{lemma:investment-conditional}
Let $B$ be an investment menu with threshold $w^*\in\mathbb R$ and complete ordering $\succeq$, and let $\beta\in B^X$ be
monotone relative to $\succeq$.
Let $\chi$ denote the signal generated
by the experiment, and define
$$
Y^1=0,
\qquad
Y^2=\beta_\chi(W)-\mu.
$$
Then versions of $F_{Y^1|W}$ and $F_{Y^2|W}$ can be chosen so as to satisfy
the hypotheses of Lemma~\ref{lemma:sc_random}\ref{item:sc_random:logconcave}.%
\end{lemma}

\begin{proof}[Proof of \Cref{lemma:investment-conditional}]
Choose versions such that
\begin{align*}
F_{Y^1|W}(y|z)&=\mathbf 1_{\mathbb R_+}(y),\\
F_{Y^2|W}(y|z)&=\sum_{x\in X}h_z(x)\mathbf 1_{\{\beta_x(z)-\mu\le y\}},
\end{align*}
and write
$$
\Delta(y,w)=\mathbf 1_{\mathbb R_+}(y)-\sum_{x\in X}h_{w-y}(x)\mathbf 1_{\{\beta_x(w-y)-\mu\le y\}}.
$$
For every $w$, $\Delta(y,w)\le0$ when $y<0$ and $\Delta(y,w)\ge0$ when $y\ge0$, so $y\mapsto\Delta(y,w)$ is single crossing.

Every nonzero asset is ranked above $\mathbf{0}$: if $\mathbf{0}\succeq b$, then $b\ge0$ below $w^*$ and $b\le0$ above $w^*$, which together with monotonicity of $b$ implies $b=\mathbf{0}$. Hence $b(z)\le0$ for every $b\in B$ and $z\le w^*$. Fix $y$. If $w\le w^*$ and $y<0$, then $\Delta(y,w)\le0$; if $y\ge0$, then $w-y\le w^*$ and every indicator equals one, so $\Delta(y,w)=0$. Thus $\Delta(y,\cdot)$ is nonpositive on $(-\infty,w^*]$.

Now fix $w^*\le w<w'$, set $z=w-y$ and $z'=w'-y$, and let $I_x(t)=\mathbf 1_{\{\beta_x(t)-\mu\le y\}}$. Since every asset is nondecreasing, $I_x(z)\ge I_x(z')$. Moreover, $x\mapsto I_x(z')$ is nonincreasing. Indeed, if $z'>w^*$, this follows from monotonicity of $\beta$ and the investment-menu ordering; if $z'\le w^*$, then $y=w'-z'>0$ and every indicator equals one. Since $z\le z'$, \Cref{lemma:tp2-fosd} yields
\begin{align*}
F_{Y^2|W}(y|z)
&=\sum_{x\in X}h_z(x)I_x(z)\\
&\ge\sum_{x\in X}h_z(x)I_x(z')\\
&\ge\sum_{x\in X}h_{z'}(x)I_x(z')
=F_{Y^2|W}(y|z').
\end{align*}
Therefore $\Delta(y,w')\ge\Delta(y,w)$. The hypotheses of \Cref{lemma:sc_random}\ref{item:sc_random:logconcave} hold with $\underline w=w^*$ and $\bar w=\infty$.
\end{proof}
\begin{lemma}\label{lemma:insurance-conditional}
Let random wealth $W$ have support $(-\infty,\omega]$, let $B$ be an insurance menu with threshold $w^*\in(-\infty,\omega)$ and witnessing complete ordering $\succeq$, and let $\beta\in B^X$ be monotone relative to the reverse of $\succeq$. Let $\chi$ denote the signal generated
by the experiment, and define
$$
Y^1=\beta_\chi(W)-\mu,
\qquad
Y^2=0.
$$
Then versions of $F_{Y^1|W}$ and $F_{Y^2|W}$ can be chosen so as to satisfy
the hypotheses of Lemma~\ref{lemma:sc_random}\ref{item:sc_random:logconvex}.%
\end{lemma}

\begin{proof}[Proof of \Cref{lemma:insurance-conditional}]
Choose versions such that
\begin{align*}
F_{Y^1|W}(y|z)&=\sum_{x\in X}h_z(x)\mathbf 1_{\{\beta_x(z)-\mu\le y\}},\\
F_{Y^2|W}(y|z)&=\mathbf 1_{\mathbb R_+}(y),
\end{align*}
and write
$$
\Delta(y,w)=\sum_{x\in X}h_{w-y}(x)\mathbf 1_{\{\beta_x(w-y)-\mu\le y\}}-\mathbf 1_{\mathbb R_+}(y).
$$
For every $w$, $-\Delta(y,w)\le0$ when $y<0$ and $-\Delta(y,w)\ge0$ when $y\ge0$, so $y\mapsto-\Delta(y,w)$ is single crossing.

Every nonzero asset is ranked above $\mathbf{0}$: if $\mathbf{0}\succeq b$, then $b\le0$ below $w^*$ and $b\ge0$ above $w^*$, which together with monotonicity of $b$ implies $b=\mathbf{0}$. Hence every $b\in B$ is nonnegative below $w^*$, nonpositive above $w^*$, and vanishes at $w^*$. Moreover, every insurance payoff is continuous: if $s<t$, then $0\le b(s)-b(t)\le t-s$ because $b$ is nonincreasing and $w\mapsto w+b(w)$ is nondecreasing.

Define
$$
\bar w=\inf\{w\in\mathbb R:b(w)\le\mu\text{ for every }b\in B\}.
$$
The set is a closed upper interval. For every $b\in B$,
$$
w^*-\mu+b(w^*-\mu)\le w^*+b(w^*)=w^*,
$$
so $b(w^*-\mu)\le\mu$ and $\bar w\le w^*-\mu$.

Fix $y$. If $w\ge\bar w$ and $y<0$, then $\Delta(y,w)\ge0$. If $y\ge0$, set $z=w-y$. Since $w\mapsto w+\beta_x(w)$ is nondecreasing and $\beta_x(w)\le\mu$,
$$
z+\beta_x(z)\le w+\beta_x(w),
$$
hence $\beta_x(z)-\mu\le y$ for every $x$. Thus all indicators equal one and $\Delta(y,w)=0$.

Now fix $w<w'<\bar w$, set $z=w-y$ and $z'=w'-y$, and let $I_x(t)=\mathbf 1_{\{\beta_x(t)-\mu\le y\}}$. Since every insurance payoff is nonincreasing, $I_x(z)\le I_x(z')$. We claim that $x\mapsto I_x(z')$ is nondecreasing. If $z'<w^*$, this follows from monotonicity of $\beta$ relative to the reverse menu order. If $z'\ge w^*$, then $y=w'-z'\le0$. Since $w'<\bar w$, choose $b_0\in B$ with $b_0(w')>\mu$. For each $x$, completeness allows us to choose $b\succeq\beta_x$ with $b(w')>\mu$: take $b=b_0$ if $b_0\succeq\beta_x$ and $b=\beta_x$ otherwise. Since $z'\ge w^*$ and $w'<w^*$,
$$
\beta_x(z')\ge b(z')
\quad\text{and}\quad
b(z')-y=z'+b(z')-w'\ge b(w')>\mu.
$$
Hence $\beta_x(z')-\mu>y$, so every indicator equals zero.

Since $z\le z'$, \Cref{lemma:tp2-fosd} gives
\begin{align*}
F_{Y^1|W}(y|z)
&=\sum_{x\in X}h_z(x)I_x(z)\\
&\le\sum_{x\in X}h_z(x)I_x(z')\\
&\le\sum_{x\in X}h_{z'}(x)I_x(z')
=F_{Y^1|W}(y|z').
\end{align*}
Thus $\Delta(y,w')\ge\Delta(y,w)$.

Finally, fix $w<\bar w$ and $y<w-\omega$, and set $\ell=w-y>\omega>w^*$. Since every asset vanishes at $w^*$ and total wealth is nondecreasing,
$$
\ell+\beta_x(\ell)-\mu\ge w^*+\beta_x(w^*)-\mu=w^*-\mu\ge\bar w>w.
$$
Thus $\beta_x(\ell)-\mu>y$ for every $x$, while $y<0$. Both conditional distribution functions therefore vanish at $(y,\ell)$, so $\Delta(y,w)=0$. All the hypotheses of \Cref{lemma:sc_random}\ref{item:sc_random:logconvex} hold with $\underline w=-\infty$.
\end{proof}
\section{Proofs of Section \ref{sec:reversals}}\label{app:reversals}

The two reversal proofs use different tail comparisons and calibrations. The investment proof uses adjacent-mass geometry and calibrates a one-parameter family of menus by the intermediate value theorem. The insurance proof uses the extreme-tail implication of super-exponentiality and sets the information price equal to the expected gross insurance payoff of a risk-neutral benchmark. A local increase in risk aversion then produces each reversal.

\subsection{Ingredients of the Investment Reversal}\label{app:adjacent-mass}
For adjacent intervals of lengths $a$ and $\mu$, $P_{a,\mu}(m)$ is the share of the mass between $m-a$ and $m+\mu$ that lies above $m$. The next lemma shows that this share increases as the intervals move right under strict log convexity and decreases under strict log concavity; both directions are also used in \Cref{sec:assumption-diagnostic}.

\begin{lemma}[Left-tail geometry]\label{lemma:left-tail-geometry} Suppose that $f$ is continuously differentiable and strictly positive on $(-\infty,\omega)$. Fix $a,\mu>0$, and define
$$
P_{a,\mu}(m)
=
\frac{F(m+\mu)-F(m)}
     {F(m+\mu)-F(m-a)}
$$
whenever $m+\mu<\omega$.

\begin{enumerate}[label=(\roman*)]
\item If $\log f$ is strictly convex on $(-\infty,\omega)$, then $P_{a,\mu}$ is strictly increasing.

\item If $\log f$ is strictly concave on $(-\infty,\omega)$, then $P_{a,\mu}$ is strictly decreasing.
\end{enumerate}
\end{lemma}
\begin{proof}
Set $A^+(m)=F(m+\mu)-F(m)$ and $A^-(m)=F(m)-F(m-a)$. Since $P_{a,\mu}=A^+/(A^++A^-)$, it has the same monotonicity as $A^+/A^-$. Writing $r=(\log f)'=f'/f$, we obtain
$$
\frac{\mathrm d}{\mathrm dm}\log\frac{A^+(m)}{A^-(m)}=\frac{\int_m^{m+\mu}r(t)f(t)\,\mathrm dt}{\int_m^{m+\mu}f(t)\,\mathrm dt}-\frac{\int_{m-a}^{m}r(t)f(t)\,\mathrm dt}{\int_{m-a}^{m}f(t)\,\mathrm dt}.
$$
The two terms are the $f(t)\,\mathrm dt$-weighted averages of $r$ over the adjacent intervals $[m,m+\mu]$ and $[m-a,m]$. If $\log f$ is strictly convex, then $r$ is nondecreasing and nonconstant on $[m-a,m+\mu]$, so the first average is strictly larger than the second. Thus $A^+/A^-$, and hence $P_{a,\mu}$, is strictly increasing. Under strict log concavity the inequalities are reversed, so $P_{a,\mu}$ is strictly decreasing.
\end{proof}

\subsection{Proof of Proposition \ref{prop:invest-log-convex}}\label{app:invest-log-convex}

The proof proceeds in three steps. First, we construct a one-parameter family of binary investment menus under a fixed extreme TP2 experiment, with the same strategy strictly optimal throughout the family and the status quo strictly optimal under the prior. Second, we calibrate the parameters so that the value of information at the prescribed price varies continuously across the family and crosses zero, allowing the asset to be selected by the intermediate value theorem. Third, Lemma~\ref{lemma:left-tail-geometry} determines the sign of the distributional comparison generated by the calibrated asset, and a local perturbation of the benchmark utility completes the proof.

\paragraph{Construction of the calibration family.}
Since $f$ is integrable on $(-\infty,\omega)$ and $\log f$ is strictly convex, there exists $w_*<\omega$ at which $f$ is locally strictly increasing. For parameters $\mu>0$ and $\bar w$ satisfying
$$
w_*<\bar w<\omega-\mu,
$$
set
$$
q
=
P_{\mu,\mu}(w_*)
=
\frac{F(w_*+\mu)-F(w_*)}
     {F(w_*+\mu)-F(w_*-\mu)}.
$$
For $\varepsilon,\delta\in(0,1)$, define
$$
v_\varepsilon(w)
=
\begin{cases}
w,&w\le \bar w,\\
\bar w+\varepsilon(w-\bar w),&w>\bar w,
\end{cases}
$$
and
$$
u_{\varepsilon,\delta}(w)
=
\begin{cases}
v_\varepsilon(w),&w\le w_*,\\
v_\varepsilon(w_*)
+\delta\bigl(v_\varepsilon(w)-v_\varepsilon(w_*)\bigr),&w>w_*.
\end{cases}
$$
Both utilities belong to $U$, and $u_{\varepsilon,\delta}$ is more risk averse than $v_\varepsilon$.

For $\widehat w<w_*-\mu$ and $L>0$, consider the binary experiment
$$
h_w(1)
=
\begin{cases}
0,&w\le\widehat w,\\
q,&w>\widehat w,
\end{cases}
\qquad
h_w(0)=1-h_w(1).
$$
The experiment is TP2.
For each $c\in[\mu,2\mu]$, define
$$
b_c(w)=
\begin{cases}
-L,&w<\widehat w,\\
0,&w=\widehat w,\\
c,&w>\widehat w.
\end{cases}
$$
Let $B_c=\{\mathbf{0},b_c\}$. Since $b_c$ is nondecreasing, nonpositive below $\widehat w$, nonnegative above $\widehat w$, and vanishes at $\widehat w$, $B_c$ is an investment menu with threshold $\widehat w$.
Consider the strategy
$$
\beta^c_0=\mathbf{0},
\qquad
\beta^c_1=b_c.
$$
\paragraph{Calibration of the parameters.}

Choose $\mu>0$ sufficiently small that $w_*+\mu<\omega$ and
$q=P_{\mu,\mu}(w_*)>\frac12.$

Consider first the auxiliary capped utility
$v_0(w)=\min\{w,\bar w\},$
where $\bar w$ will be selected below. At $c=2\mu$, first consider the limiting experiment obtained by sending $\widehat w$ to $-\infty$. Under the strategy $\beta^{2\mu}$, terminal wealth then equals $W+\mu$ with probability $q$ and $W-\mu$ with probability $1-q$. Its distribution $\widetilde G$ satisfies
$$
\widetilde G(m)
=
qF(m-\mu)+(1-q)F(m+\mu).
$$
For every $m$ such that $m+\mu<\omega$,
$\widetilde G(m)>F(m)$ iff $q<P_{\mu,\mu}(m).$
Since $q=P_{\mu,\mu}(w_*)$ and $P_{\mu,\mu}$ is strictly increasing, $\widetilde G(m)<F(m)$ for every $m<w_*$, while $\widetilde G(m)>F(m)$
for every $w_*<m<\omega-\mu$.
In particular,
$\int_{-\infty}^{w_*}\bigl(F(m)-\widetilde G(m)\bigr)\,\mathrm dm>0.
$

Choose $\bar w\in(w_*,\omega-\mu)$ sufficiently close to $w_*$ that $\int_{-\infty}^{\bar w}
\bigl(F(m)-\widetilde G(m)\bigr)\,\mathrm dm>0.
$
Integration by parts gives $\int v_0\,d(\widetilde G-F)>0.$

We next choose the cutoff $\widehat w$. Let $G^{\widehat w}_{2\mu}$ denote the terminal-wealth distribution induced by the actual experiment and the strategy $\beta^{2\mu}$. The actual and limiting constructions differ only when $W\le\widehat w$. Since $v_0$ is $1$-Lipschitz,
$$
\left|
\int v_0\,\mathrm dG^{\widehat w}_{2\mu}
-
\int v_0\,d\widetilde G
\right|
\le
2q\mu F(\widehat w).
$$
Choose $\widehat w<w_*-\mu$ sufficiently far to the left that $\int v_0\,d(G^{\widehat w}_{2\mu}-F)>0.$

Having fixed $\widehat w$, choose $L>0$ sufficiently large that
$$
-LF(\widehat w)+2\mu\bigl(1-F(\widehat w)\bigr)<0.
$$
For $v_\varepsilon$, the utility gain from $b_c$ equals $-L$ on $(-\infty,\widehat w)$, equals zero at $\widehat w$, and is at most $2\mu$ on $(\widehat w,\infty)$.
Thus the displayed bound makes the status quo strictly optimal under the prior, uniformly over $c\in[\mu,2\mu]$ and $\varepsilon\in[0,1]$. The unnormalized gain after signal $0$ is bounded above by $-LF(\widehat w)+2\mu(1-q)(1-F(\widehat w))$, which is smaller and therefore negative. After signal $1$, the posterior is supported on $(\widehat w,\infty)$, where $b_c=c>0$, so $b_c$ is strictly optimal. Hence $\beta^c$ is the unique optimal informed strategy throughout the family.

For $c\in[\mu,2\mu]$, let $G_c$ denote the terminal-wealth distribution induced by $\beta^c$, and define
$S_\varepsilon(c)
=
\int v_\varepsilon\,d(G_c-F).$
At $c=\mu$, purchasing information leaves wealth unchanged after signal $1$ and lowers it by $\mu$ after signal $0$. Hence $S_0(\mu)<0.$
By the choice of $\widehat w$, $S_0(2\mu)>0.$
Continuity in $c$ therefore allows us to choose
$\mu<c_-<c_+<2\mu$ such that $S_0(c_-)<0<S_0(c_+).$

We now choose $\varepsilon$. For every $c\in[c_-,c_+]$ and every $m\in[w_*,\bar w]$, one has
$m+\mu-c\ge w_*-\mu>\widehat w.$
Hence
$G_c(m)=qF(m+\mu-c)+(1-q)F(m+\mu).$
Consequently, $G_c(m)>F(m)$ iff $q<P_{c-\mu,\mu}(m).$

Since $c<2\mu$, one has $c-\mu<\mu$. The map $a\mapsto P_{a,\mu}(m)$ is strictly decreasing, and Lemma~\ref{lemma:left-tail-geometry}(i) gives
$P_{c-\mu,\mu}(m)> P_{\mu,\mu}(m)\ge P_{\mu,\mu}(w_*) =
q.$

Thus $G_c(m)>F(m)$
throughout $[w_*,\bar w]\times[c_-,c_+]$. By continuity and compactness,
$$
A
:=
\min_{c\in[c_-,c_+]}
\int_{w_*}^{\bar w}
\bigl(G_c(m)-F(m)\bigr)\,\mathrm dm
>0.
$$
Moreover, since the terminal wealth induced by $\beta^c$ has uniformly bounded first moments over $c\in[c_-,c_+]$,
$$
K
:=
\sup_{c\in[c_-,c_+]}
\int_{\bar w}^{\infty}
|G_c(m)-F(m)|\,\mathrm dm
<\infty.
$$
Choose $\varepsilon>0$ sufficiently small that $\varepsilon K<A$ and
$S_\varepsilon(c_-)<0<S_\varepsilon(c_+).$

Since $c\mapsto b_c$ is continuous in $L^1(F)$ and $v_\varepsilon$ is Lipschitz continuous, $S_\varepsilon$ is continuous on $[c_-,c_+]$. Hence there exists $c^*\in(c_-,c_+)$ such that
$S_\varepsilon(c^*)=0.$

For the menu $B_{c^*}$, the status quo is strictly optimal under the prior and $\beta^{c^*}$ is the unique optimal informed strategy. Hence $M(v_\varepsilon)=\mu$.

It remains to choose $\delta$. Write $G=G_{c^*}$ and $v=v_\varepsilon$. Since
$u_{\varepsilon,\delta}(w)
= v(w) + (\delta-1)\bigl(v(w)-v(w_*)\bigr)^+,
$ integration by parts gives
\begin{align*}
\int \bigl(v-v(w_*)\bigr)^+\,d(G-F)
&=
-\int_{w_*}^{\bar w}\bigl(G(m)-F(m)\bigr)\,\mathrm dm\\
&\qquad
-\varepsilon\int_{\bar w}^{\infty}
\bigl(G(m)-F(m)\bigr)\,\mathrm dm\\
&\le
-A+\varepsilon K
<0.
\end{align*}
Since $\int v\,d(G-F)=0$, it follows that, for every $\delta<1$,
$\int u_{\varepsilon,\delta}\,d(G-F)>0.
$
Finally, choose $\delta<1$ sufficiently close to $1$ that the strict prior optimality of the status quo is preserved, and set $v=v_\varepsilon$ and $u=u_{\varepsilon,\delta}$. Then both utilities satisfy \eqref{eq:ZAZI}, $u$ is more risk averse than $v$, and $M(v)=\mu$. The strategy $\beta^{c^*}$ yields strictly positive information surplus for $u$ at price $\mu$, so \Cref{prop:mustar} gives $M(u)>\mu=M(v)$. Since $B_{c^*}$ is an investment menu and $\alpha$ is TP2, this proves \Cref{prop:invest-log-convex}.
\subsection{Proof of Proposition~\ref{prop:insurance-super-exponential}}
\label{app:insurance-super-exponential}

We first note that super-exponentiality implies that
$\lim_{m\to-\infty}F(m+\Delta)/F(m)=+\infty$ for every $\Delta>0$.
Indeed, fix $K>0$. For all sufficiently negative $m$, $(\log f)'(s)\ge K$ whenever $s\le m+\Delta$. Hence $f(s+\Delta)\ge e^{K\Delta}f(s)$ for every $s\le m$. Integrating with respect to $s$ over $(-\infty,m]$ gives $F(m+\Delta)\ge e^{K\Delta}F(m)$. Since $K$ is arbitrary, the claim follows.

We next construct the benchmark decision problem. Fix $q\in(0,1)$ and $a,\ell>0$. For $z\in\mathbb R$, define
$b_z(w)=\max\{-\ell,\min\{a,z-w\}\}$. The payoff $b_z$ is nonincreasing, $w+b_z(w)$ is nondecreasing, and $b_z$ is nonnegative below $z$ and nonpositive above $z$. Thus $B_z=\{\mathbf{0},b_z\}$ is an insurance menu with threshold $z$. Moreover, $b_z(w)\to-\ell$ as $z\to-\infty$ for every $w$, while $|b_z|\le\max\{a,\ell\}$. Hence, by dominated convergence, $\int b_z\,\mathrm dF\to-\ell$. Choose $z$ sufficiently negative that $\int b_z\,\mathrm dF<0$, write $b=b_z$ and $B=B_z$, and choose $\widehat w<\min\{z-a,\omega\}$.

Let $X=\{0,1\}$ and define the experiment $\alpha$ by $h_w(0)=q\mathbf 1_{\{w\le\widehat w\}}$ and $h_w(1)=1-h_w(0)$. Since $w\mapsto h_w(0)$ is nonincreasing, this binary experiment is TP2. Set $p=F(\widehat w)$. Because $\widehat w<\omega$ and $f$ is strictly positive on $(-\infty,\omega)$, $p\in(0,1)$. Hence both signals have positive ex ante probability. Let $v(w)=w$, set $\mu=pqa\in(0,a)$, and consider the strategy $\beta_0=b$ and $\beta_1=\mathbf{0}$.

Because $v$ is linear, the common information price does not affect the action comparisons. Then, at the prior, the expected net gain of choosing $b$ is $\int b\,\mathrm dF<0$, so that $v$ satisfies \eqref{eq:ZAZI}. Since the expected net gain after signal $0$ is $a > 0$, the expected net gain after signal $1$ is strictly negative and, hence, $\beta$ is optimal for $v$. Since signal $0$ has ex-ante probability $pq$, it follows that the expected net gain of acquiring information (and acting optimally) is $pqa-\mu = 0$, so that $M(v)=\mu$.

We now identify the extreme-tail deterioration that greater risk aversion will amplify. Let $G=G^\mu_\beta$. For every $m<\widehat w-\mu$,
$$
G(m)=qF(m+\mu-a)+(1-q)F(m+\mu),
\qquad
\frac{G(m)}{F(m)}
\ge
(1-q)\frac{F(m+\mu)}{F(m)}
\longrightarrow+\infty.
$$
Thus there exists $t<\widehat w-\mu$ such that $G(m)-F(m)>0$ for every $m<t$. Since $F$ and $G$ have finite first moments, $I:=\int_{-\infty}^t(G(m)-F(m))\,\mathrm dm$ is finite and strictly positive.

For $\lambda>1$, define $u_\lambda(w)=w-(\lambda-1)(t-w)^+$. This utility is increasing, Lipschitz continuous, and concave, hence more risk averse than $v$. Since $F$ and $G$ have the same mean and $(t-w)^+=\int_{-\infty}^t\mathbf 1_{\{w\le m\}}\,\mathrm dm$, integration by parts gives
$$
\int u_\lambda\,\mathrm d(G-F)
=
-(\lambda-1)\int_{-\infty}^t\bigl(G(m)-F(m)\bigr)\,\mathrm dm
=
-(\lambda-1)I<0.
$$
At $\lambda=1$, the prior action comparison and both posterior action comparisons are strict. Their expected-utility differences are continuous in $\lambda$ by dominated convergence, so the same signs hold for some $\lambda>1$ sufficiently close to one. For $u=u_\lambda$, the status quo remains uniquely optimal under the prior and $\beta$ remains the unique optimal informed strategy at price $\mu$. Thus $u$ satisfies \eqref{eq:ZAZI} but strictly prefers not to purchase information at $\mu=M(v)$, and \Cref{prop:mustar} yields $M(u)<M(v)$.
\section{Proofs of Lemmas Used in Examples}
\begin{lemma}[Two-exponential ratio]\label{lemma:two-exponential-ratio}
Let $a,b>0$, $c\geq0$, and $0\leq\lambda_0<\lambda_1$. For $r>0$, set $X_\lambda(r)=ae^{-cr}+be^{-\lambda r}$ and
$$
j(r)=-\frac{1}{r}\log\frac{X_{\lambda_1}(r)}{X_{\lambda_0}(r)}.
$$
If $\lambda_1\leq c$, then $j$ is strictly increasing; if $c\leq\lambda_0$, then $j$ is strictly decreasing.
\end{lemma}

\begin{proof}
Set $\Phi_\lambda(r)=\log X_\lambda(r)-r\partial_rX_\lambda(r)/X_\lambda(r)$. Direct differentiation gives
$$
j'(r)=\frac{\Phi_{\lambda_1}(r)-\Phi_{\lambda_0}(r)}{r^2},
\qquad
\partial_\lambda\Phi_\lambda(r)=\frac{ab r^2(c-\lambda)e^{-(c+\lambda)r}}{X_\lambda(r)^2}.
$$
Thus $\Phi_\lambda(r)$ is strictly increasing in $\lambda$ below $c$ and strictly decreasing above $c$. Since $\lambda_0<\lambda_1$, integrating over $[\lambda_0,\lambda_1]$ gives the two strict conclusions.
\end{proof}

\begin{proof}[Proof of \Cref{lemma:example1}]
Let $q=1-p$ and $d=k_H-k_L$. Factoring $e^{-rk_L}$ shows that the ratio inside the logarithm is $X_{d+H}(r)/X_d(r)$ for $X_\lambda(r)=p+qe^{-\lambda r}$. Apply \Cref{lemma:two-exponential-ratio} with $a=p$, $b=q$, $c=0$, $\lambda_0=d$, and $\lambda_1=d+H$.
\end{proof}

\begin{proof}[Proof of \Cref{lemma:example2}]
Let $q=1-p$ and $\Delta=k_H-k_L$. Factoring $e^{-rk_L}$ shows that the ratio inside the logarithm is $X_L(r)/X_0(r)$ for $X_\lambda(r)=qe^{-\Delta r}+pe^{-\lambda r}$. Since $0<L<\Delta$, apply \Cref{lemma:two-exponential-ratio} with $a=q$, $b=p$, $c=\Delta$, $\lambda_0=0$, and $\lambda_1=L$.
\end{proof}

\begin{proof}[Proof of \Cref{lemma:example3}]
Set $\delta=k_H-k_L$ and, for $a,r>0$, let $\varphi_r(a)=a/(e^{ar}-1)$. Then
$$
\frac{\mathrm d}{\mathrm dr}\log \rho(r)=-\delta+\varphi_r(H_E)-\varphi_r(L_P).
$$
The inequalities $1-e^{-ar}<ar<e^{ar}-1$ imply
$$
\frac{e^{-ar}}{r}<\varphi_r(a)<\frac{1}{r}.
$$
Hence
$$
\varphi_r(H_E)-\varphi_r(L_P)<\frac{1-e^{-L_Pr}}{r}<L_P,
$$
so $(\log \rho)'(r)<-\delta+L_P<0$. Since $\rho>0$, it is strictly decreasing. Moreover,
$$
\lim_{r\downarrow0}\rho(r)=\frac{1-p}{p}\frac{H_E}{L_P},
\qquad
\lim_{r\to\infty}\rho(r)=0.
$$
If the first limit is at most one, strict decrease gives $\rho(r)<1$ for every $r>0$. If it exceeds one, continuity and strict decrease give a unique $r_0>0$ such that $\rho(r_0)=1$, with $\rho(r)>1$ for $r<r_0$ and $\rho(r)<1$ for $r>r_0$.
\end{proof}

\end{appendices}

\end{document}